\documentclass[letterpaper]{article}

\usepackage{hyperref}

\usepackage[utf8]{inputenc}
\usepackage{algorithm}
\usepackage{algorithmicx}
\usepackage{algcompatible}
\usepackage{multicol}
\usepackage[margin=1in]{geometry}
\usepackage{graphicx}
\usepackage{float}
\usepackage{subfig}
\usepackage{amsthm,color}
\usepackage{amsmath,amssymb}
\usepackage[T1]{fontenc} 
\usepackage[shortlabels]{enumitem}
\usepackage[title]{appendix}
\usepackage{thmtools,thm-restate}

\newcommand{\rmv}[1]{}

\newcommand{\MP}[1]{}
\newcommand{\RMP}[1]{}

\newcommand{\XMP}[1]{}

\newcommand{\VMP}[1]{}

\newcommand{\REG}{\textbf{R}}

\newcommand{\nts}{\textit{new\_ts}}
\newcommand{\nsq}{\textit{new\_sq}}

\newcommand{\str}{\textrm{write strong-linearization}} 
\newcommand{\sly}{\textrm{write strongly-linearizable}}

\newcommand{\RS}{\textit{S}_R}
\newcommand{\W}[1]{w_{#1}}
\newcommand{\RD}[1]{r_{#1}}

\newcommand{\lS}{\textit{S}}
\newcommand{\ts}[1]{\textit{ts}_{#1}}

\newcommand{\seth}{\mathcal{H}}

\newcommand{\ex}{\textrm{reach}}
\newcommand{\exs}{\textrm{reaches}}
\newcommand{\twkp}{t_w^{k+1}}

\newcommand{\trk}{t_r^{k}}
\newcommand{\two}{t_w^{0}}

\newcommand{\swA}{\mathcal{A}}

\newcommand{\seq}[1]{S_{#1}}
\newcommand{\wseq}[1]{WS_{#1}}
\newcommand{\sC}[1]{\mathcal{C}_{#1}}
\newcommand{\B}[1]{\mathcal{B}_{#1}}
\newcommand{\pts}[2]{\textit{ts}_{#1}^{#2}}
\newcommand{\SR}{S_\mathcal{R}}

\newcommand{\Awsl}{\textrm{Algorithm}~\ref{wsl}}
\newcommand{\f}{\textrm{Algorithm}~\ref{f}}
\newcommand{\Al}{\textrm{Algorithm}~\ref{Al}}

\newcommand{\procs}{p_2, p_3, \ldots, p_{n-1}}
\newcommand{\Procs}{\textrm{processes } p_2, p_3, \ldots, p_{n-1}}

\newcounter{theorem}
\newcounter{NewCounter}
\newcounter{claimcount}[NewCounter]
\makeatletter
\renewcommand{\p@claimcount}{\theNewCounter.}
\makeatother

\newcounter{cclaimcount}[claimcount]
\makeatletter
\renewcommand{\p@cclaimcount}{\theNewCounter.\theclaimcount.}
\makeatother

\newenvironment{claim}{
\par\refstepcounter{claimcount}\vspace{5pt}\noindent\textbf{Claim \arabic{NewCounter}.\arabic{claimcount}}\begin{itshape}
}{\end{itshape}}

\newenvironment{cclaim}{ 
\par\refstepcounter{cclaimcount}\vspace{5pt}\noindent\textbf{Claim \arabic{NewCounter}.\arabic{claimcount}.\arabic{cclaimcount}}\begin{itshape}}{\end{itshape}}

\newtheorem{theorem}[NewCounter]{Theorem}

\newtheorem{corollary}[NewCounter]{Corollary}
\newtheorem{lemma}[NewCounter]{Lemma}
\newtheorem{observation}[NewCounter]{Observation}

\newtheorem*{notation}{Notation}
\newtheorem{definition}[NewCounter]{Definition}

\title{On Register Linearizability and Termination}

\author{Vassos Hadzilacos \qquad Xing Hu \qquad Sam Toueg\\~\\Department of Computer Science\\University of Toronto\\Canada}

\begin{document}
\maketitle

\begin{abstract}
 It is well-known that, for deterministic algorithms,
 	\emph{linearizable} objects can be used
	as if they were atomic objects.
As pointed out by Golab, Higham, and Woelfel, however,
	a \emph{randomized} algorithm that works with atomic objects may lose some of its properties
	if we replace the atomic objects that it uses with objects that are only linearizable.
It was not known whether the properties that can be lost include the all-important property of termination (with probability 1).
In this paper, we first show that a randomized algorithm
	can indeed lose its termination property if we replace the atomic registers that it uses
	with linearizable ones.

Golab \emph{et al.} also introduced \emph{strong linearizability}, 
	and proved that strongly linearizable objects can be used
	as if they were atomic objects, \emph{even for randomized algorithms}:
	they can replace atomic objects while preserving the algorithm's correctness properties,
	including termination.
Unfortunately, there are important cases where strong linearizability is impossible to achieve.
In particular, Helmi, Higham, and Woelfel showed a large class of ``non-trivial'' objects, including MWMR registers,
	do not have strongly linearizable implementations from SWMR registers. 

Thus we propose a new type of register linearizability, called \emph{write strong-linearizability},
	that is strictly stronger than (plain) linearizability
	but strictly weaker than strong linearizability.
This intermediate type of linearizability has some desirable properties.
We prove that some randomized algorithms that
	fail to terminate with linearizable registers,	
	work with $\sly$ ones.
In other words, 
	there are cases where linearizability is not sufficient but
	write strong-linearizability is.
In contrast to the impossibility result mentioned above, 
	we prove that $\sly$ MWMR registers \emph{are} implementable from SWMR registers.
	Achieving write strong-linearizability, however, is harder than achieving just linearizability:
	we give a simple implementation of  MWMR registers
	from SWMR registers
	and we prove that this implementation is linearizable but \emph{not} $\sly$.
Finally, 
	we prove that \emph{any} linearizable implementation of SWMR registers is necessarily $\sly$;
	this holds for shared-memory, message-passing, and hybrid systems.
\end{abstract}


\section{Introduction}

 \emph{Linearizability} is a well-known and very useful property of shared object implementations~\cite{linearizability}.
Intuitively, with a linearizable (implementation of) object each operation must appear as if it takes effect instantaneously
	at some point during the time interval that it actually spans;
	for deterministic algorithms linearizable objects can be used as if they were atomic.\footnote{Throughout the paper we consider only implementations that are \emph{wait free}~\cite{herlihy91}.}
As pointed out by the seminal work of Golab \emph{et al.} \cite{sl11}, however,
	linearizable objects are not as strong as atomic objects in the following sense:
	a randomized algorithm that works with atomic objects may lose some of its properties
	if we replace the atomic objects that it uses with objects that are only linearizable.
In particular, they present a randomized algorithm
	that guarantees that
	some random variable has \emph{expected value} 1,
	but if we replace the algorithm's atomic registers with linearizable registers,
	a \mbox{\emph{strong adversary}} can manipulate schedules to
	ensure that this random variable
	has expected value $\frac{1}{2}$.

\subsection{Linearizability and termination}
A natural question is whether termination is one of the properties that can be lost
	with the use of linearizable objects.
More precisely:
	is there a randomized algorithm that 
	(a) terminates with probability 1 against a strong adversary when the objects that it uses are atomic, but
	(b) when these objects are replaced with
	linearizable objects, 
 	a strong adversary can 
	ensure that
	the algorithm never terminates?
This question is particularly interesting because
	achieving termination is
	one of the main uses of randomized algorithms
	(e.g., to ``circumvent'' the famous FLP impossibility result~\cite{flp})~\cite{abrahamson1988,aspnes1993,aspnes1998,
	aspnes2003,aspnes1990,aspnes1992,attiya08,bracha1991,chandra1996},
	but it
	is not answered by the results in~\cite{sl11},
	and to the best of our knowledge,
	it is also not addressed in subsequent papers on this subject~\cite{sl19,sl15,sl12}.


In this paper, among other things, we show that termination can be lost.
Specifically, we give a randomized algorithm 
	that uses registers such that:
	(1) if these registers are atomic then the algorithm terminates,
	i.e., all the processes 
	halt, with probability~1, even against a strong adversary,
	and
	(2) if the registers are ``only'' linearizable,
	a strong adversary can prevent termination:
	it can always manipulate schedules so that processes never halt. 
%

\subsection{Strong linearizability}
Golab \emph{et al.} also introduced a stronger version of linearizability called \emph{strong linearizability}~\cite{sl11}.
Intuitively, while in linearizability the order of all operations can be determined ``off-line'' given the entire execution,
	in strong linearizability the order of all operations has to be fixed irrevocably ``on-line'' without knowing the rest of the execution.
	Golab \emph{et al.} proved that strongly linearizable (implementations of) objects are ``as good'' as atomic objects, even for randomized algorithms:
	they can replace atomic objects while preserving the algorithm's correctness properties
	including termination with probability~1.
Unfortunately, there are important cases where strong linearizability is impossible to achieve.

For example, Helmi \emph{et al.} proved that
	a large class of so-called \emph{non-trivial} objects, including multi-writer multi-reader (MWMR) registers,
	do not
	have strongly linearizable implementations from single-writer multi-reader (SWMR) registers~\cite{sl12}.
This impossibility result may affect many existing randomized algorithms (e.g.,\cite{renaming2014,renaming2011,CRenaming2011,renaming2010,aspnes1993,appcount10}):
,	these algorithms use atomic MWMR registers,
	so if we want to run them in systems with SWMR registers
	we cannot hope to \emph{automatically} do so
	just by replacing their atomic MWMR registers
	with strongly linearizable implementations from SWMR registers.
%

Similarly, consider the well-known ABD algorithm that implements linearizable SWMR
	registers in message-passing systems~\cite{abd}.\footnote{This implementation works under the assumption that fewer than half of the processes may crash.}
One important use of this algorithm is to relate message-passing and shared-memory systems as follows:
	any algorithm that works with atomic shared registers can be automatically
	transformed into an algorithm for message-passing systems
	by replacing its atomic registers with the ABD register implementation.
It has been recently shown, however, 
	that the ABD register implementation is \emph{not} strongly linearizable~\cite{abdnotsl}.
Thus one cannot use the ABD implementation to automatically transform any shared-memory randomized algorithm
	that terminates with probability~1
	into an algorithm that works in message-passing systems:
	using the ABD register implementation instead of atomic registers
	in a randomized algorithm may prevent termination.


\subsection{Write strong-linearizability}

Motivated by the impossibility of implementing strongly linearizable registers mentioned above,
	we propose a new type of register linearizability, called \emph{write strong-linearizability},
	that is strictly stronger than (plain) linearizability
	but strictly weaker than strong linearizability.
Intuitively, while in strong linearizability the order of \emph{all} operations has to be fixed irrevocably ``on-line''
	without knowing the rest of the execution,
	in write strong-linearizability only the write operations must be ordered ``on-line''.
This intermediate type of linearizability has some desirable properties,
	as described below:
	
%
%
%

\begin{itemize}

\item In some cases where linearizability is not sufficient to achieve termination,
	write strong-linearizability is.
	To show this, we describe a randomized algorithm such that
	if the registers of this algorithm are only \emph{linearizable},
    then a strong adversary can prevent its termination;
	but if they are \emph{$\sly$}, 
    then the algorithm terminates with probability 1 (Section~\ref{sectiontoyalgo}).
   	
   We then generalize this result as follows:
	for every randomized algorithm $\mathcal{A}$ that solves a task $T$
	and terminates with probability~1 against a strong adversary,
	there is a corresponding randomized algorithm $\mathcal{A}'$ for task $T$ such that:
	(a)~if the registers that $\mathcal{A}'$ uses are only linearizable,
		$\mathcal{A}'$ does not terminate, but
	(b) if they are $\sly$, $\mathcal{A}'$ terminates.



\item In contrast to the impossibility result proved in~\cite{sl12},
	$\sly$ MWMR registers \emph{are} implementable from SWMR registers.
	To prove this we modify a known implementation of MWMR registers~\cite{swlamport},
	and we linearize the write operations ``on-line'' by using vector timestamps
	that may be \emph{only partially formed} (Section~\ref{sectionwslalgo}).
%
	
	Achieving write strong-linearizability, however, is harder than achieving just linearizability.
	We give a simpler implementation of  MWMR registers
	from SWMR registers
	that uses Lamport clocks to timestamp writes~\cite{lam86},
	and we prove that this implementation is linearizable but \emph{not} $\sly$ (Section~\ref{sectionlamport}).

\item	Although the ABD implementation of SWMR registers is not strongly linearizable,
	we show that it is actually $\sly$.
	In fact, we prove that \emph{any} linearizable implementation of SWMR registers is necessarily $\sly$;
	this holds for message-passing, shared-memory, and hybrid systems (Section~\ref{swsection}).
\end{itemize}

Finally, it is worth noting that even though we focus on registers here,
	our intermediate notion of linearizability
	can be extended to other types of objects
	and operations.
Intuitively, an implementation of an object is \emph{strongly linearizable with respect to a subset of operations $O$}
	if the order of all operations \emph{in $O$} must be fixed irrevocably ``on-line'' without knowing the rest of the execution.

All the results presented are proven in this paper, but due to the space limitation
	several proofs are relegated to optional appendices.



\section{Model sketch}
We consider a standard distributed system where asynchronous processes that may fail by crashing
communicate via registers and other shared objects.
In such systems, shared objects can be used to \emph{implement} other shared objects
such that the implemented objects are \emph{linearizable} and \emph{wait-free}~\cite{herlihy91,linearizability}.

\subsection{Atomic registers}
A register $R$ is \emph{atomic} if its read and write operations are \emph{instantaneous} (i.e., indivisible);
	each read 
	must return the value of the last write 
	that precedes it, or the initial value of $R$ if no such write exists.
A SWMR register $R$ is shared by a set $S$ of processes such that it can be written
	(sequentially) by exactly one process $w \in S$
	and can be read by all processes~in~$S$; we say that $w$ is the \emph{writer} of $R$~\cite{lam86}.	
A MWMR register $R$ is shared by a set $S$ of processes such that it can be written
	and read by all processes~in~$S$.

\subsection{Linearizable implementations of registers}
In an object implementation,
	each operation spans an interval that starts
	with an \emph{invocation} and terminates with a \emph{response}.

\begin{definition}
Let $o$ and $o'$ be any two operations.
\begin{itemize}

\item $o$ \emph{precedes} $o'$ if the response of $o$ occurs 
	before \mbox{the invocation of $o'$.}\XMP{make terms consistent}

\item $o$ \emph{is concurrent with} $o'$ if neither precedes the other.

\end{itemize}
\end{definition}


Roughly speaking,
	an object implementation is \emph{linearizable}~\cite{linearizability}
	if, although operations can be concurrent,
	operations behave as if they occur in a \emph{sequential} order (called ``linearization order'')
	that is consistent with the order in which operations actually occur:
	if an operation $o$ precedes an operation $o'$, then $o$ is before $o'$ in the linearization order
	(the precise definition is given in~\cite{linearizability}).

Let $\seth$ be the set of histories of a register implementation.
An operation $o$ is \emph{complete} in a history $H\in\seth$ if $H$ contains both the invocation and response of $o$,
	otherwise $o$ is \emph{pending}.
\XMP{define set of history H, we have it before?}

\begin{definition}\label{defl}
A function $f$ is \emph{a linearization function for $\seth$} (with respect to the type register)
	if it maps each history $H \in \seth$ to a sequential history $f(H)$ such that:\MP{sequential history def?}
\begin{enumerate}
	\item \label{p1} $f(H)$ contains all completed operations of $H$ and possibly some non-completed ones (with matching responses added).

	\item \label{p2} If operation $o$ precedes $o'$ in $H$,
		then $o$ occurs before $o'$ in $f(H)$.

	\item \label{p3} For any read operation $r$ in $f(H)$,
		if no write operation occurs before $r$ in $f(H)$, then $r$ reads the initial value of the register; 
		otherwise, $r$ reads the value written by the last write operation that occurs before $r$ in $f(H)$.
\end{enumerate}
\end{definition}

\begin{definition}\cite{sl11}
A function $f$ is \emph{a strong linearization function for $\seth$} if:

(L) $f$ is a linearization function for $\seth$, and
	
(P) for any histories $G,H \in \seth$, if $G$ is a prefix of $H$,
	then $f(G)$ is a prefix of $f(H)$.
\end{definition}

By restricting the strong linearization requirement, i.e., property (P)
	 to write operations only,
	 we define the following:
	 
\begin{definition}\label{defwsl}
A function $f$ is \emph{a $\str$ function for $\seth$} if:

(L) $f$ is a linearization function for $\seth$, and

(P) for any histories $G,H \in \seth$, if $G$ is a prefix of $H$, 
	then the sequence of write operations in $f(G)$ is a prefix of the sequence of write operations in $f(H)$.
\end{definition}

\begin{definition}\label{L-I}
An algorithm $\mathcal{A}$ that implements a \emph{register}
	is linearizable,
	$\sly$, or strong linearizable,
	if there is a linearization, $\str$, or strong linearization function (with respect to the type register)
	for the set of histories $\seth$ of $\mathcal{A}$.
\end{definition}

\section{Termination under linearizability and write strong-linearizability}\label{sectiontoyalgo}
In this section, 
	we show that in some cases linearizability is not sufficient for termination but
	write strong-linearizability~is.
To do so,
	we present a randomized algorithm, namely Algorithm~\ref{toyalgo}, and prove that
	(a) it fails to terminate if its registers are only linearizable	
	but (b) it terminates if they are $\sly$.
We then use Algorithm~\ref{toyalgo} to show that
	every randomized algorithm $\mathcal{A}$ that solves a task
	and terminates with probability~1 against a strong adversary,
	has a corresponding randomized algorithm $\mathcal{A}'$ for the same task such that:
	(a) if the registers that $\mathcal{A}'$ uses are linearizable,
		$\mathcal{A}'$ does not terminate, but
	(b) if they are $\sly$, $\mathcal{A}'$ terminates.

\begin{algorithm}[!ht]
\caption{A game for $n\ge3$ processes}
\begin{flushleft}
\textsc{Shared MWMR registers} $R_1, R_2, C$
\vspace*{-2mm}
\end{flushleft}
\begin{multicols}{2}
\label{toyalgo}
\begin{algorithmic}[1]

\STATEx \textsc{Code of process $p_i$, $i\in \{0, 1\}$}:
\vspace{1mm}
\FOR  {rounds $j=1,2,...$}\label{enter1}
\vspace{.7mm}
\STATE \{* \textbf{Phase 1} *\}
\STATE $R_1 \gets [i,j]$ \label{pwrite1}
\IF{$i=0$}\label{cointosser1} 
\STATE \{* $p_0$ flips a coin and writes it into $C$ *\}
\STATE $c \gets$ coin flip \label{pcoin1-1}
\STATE $C \gets c$  \label{pcoin1-2}
\ENDIF
\vspace{.7mm}
\STATE \{* \textbf{Phase 2} *\}
\STATE $R_2 \gets \textsc{0}$ \label{cleanr2}
\STATE $v\gets R_2$ \label{v1}
\IF{$v < n-2$}\label{guard2}
     \STATE \textbf{exit for loop} \label{exit2}
\ENDIF
\vspace{.7mm}
\ENDFOR
\STATE \textbf{return}\label{halt0}

\columnbreak

\STATEx \textsc{Code of process $p_i$, \mbox{$i\in \{2, 3, \ldots, n-1\}$}:}
\vspace{1mm}

\FOR  {rounds $j=1,2,...$}\label{enter2}
\vspace{.7mm}
\STATE \{* \textbf{Phase 1} *\}
\STATE $R_1 \gets \bot$\label{cleanr1}
\STATE $C \gets \bot$\label{cleanc}
\STATE $u_1 \gets R_1$ \label{u1}
\STATE $u_2 \gets R_1$ \label{u2}
\STATE $c \gets C$ \label{rcoin1}
\IF{($u_1 = \bot$ \textbf{or} $u_2 = \bot$ \textbf{or} $c = \bot$)}\label{guard0}
     \STATE \textbf{exit for loop} \label{exit0}
\ENDIF
%
\IF{($u_1 \neq [c,j]$ \textbf{or} $u_2 \neq [1-c,j]$)}\label{guard1}
     \STATE \textbf{exit for loop} \label{exit1}
\ENDIF\label{endexit}
\vspace{.7mm}
\STATE \{* \textbf{Phase 2} *\}
\STATE $R_2 \gets \textsc{0}$ \label{cleanr2-2}
\STATE $v \gets R_2$ \label{pread2}
\STATE $v \gets v+1$ \label{inc2}
\STATE $R_2 \gets v$ \label{pwrite2}
\vspace{.7mm}
\ENDFOR
\STATE \textbf{return}\label{halt1}
\end{algorithmic}
\end{multicols}
\vspace*{-3mm}
\end{algorithm}

Algorithm~\ref{toyalgo} uses three MWMR registers $R_1$, $R_2$, and $C$.
It can be viewed as a game executed by $n \ge 3$ processes, which are partitioned into
	two groups: the ``hosts'' $p_0$ and $p_1$ and the ``players'' $p_2,...,p_n$.
The game proceeds in asynchronous rounds, each round consisting of two phases.
In Phase 1 of a round $j$, process $p_1$ writes $[1,j]$ into $R_1$;
	while $p_0$ first writes $[0,j]$ in $R_1$,
	 and then it writes the result
	of a 0-1 coin flip into $C$ (lines~\ref{pwrite1}-\ref{pcoin1-2}).
After doing so, each of $p_0$ and $p_1$ proceeds to Phase 2.
	
In Phase 1, each player $p_i$ ($2 \le i \le n-1)$ reads $R_1$ twice (lines~\ref{u1}--\ref{u2}),
	and then it reads $C$ (line~\ref{rcoin1}).
If $p_i$ reads $c \in \{0,1\}$ from $C$ and the sequence of two values that it read from $R_1$ is
	$[c,j]$ and then $[1-c,j]$, $p_i$ proceeds to Phase~2,
	otherwise it exits the game (lines~\ref{guard0}-\ref{endexit}).
Every player $p_i$ that stays in the game
	resets $R_2$ to $0$ (line~\ref{cleanr2-2})
	and tries to increment it by 1 (lines~\ref{pread2}-\ref{pwrite2});
	thus $R_2$ holds a lower bound on the number of players that enter Phase 2.
After doing so $p_i$ proceeds to the next round.

In Phase 2, each host $p_0$ and $p_1$ first resets $R_2$ to $0$ (line~\ref{cleanr2}) 
	and then reads $R_2$ (line~\ref{v1}).
If a host sees that $R_2 \ge n-2$
	then it is certain that \emph{all} the players remained in the game,
	and so it also remains in the game by proceeding to the next round;
	otherwise it exits the game.

We will show that if the registers are only linearizable,
 	then a strong adversary $\mathcal{S}$ can manipulate schedules such
	that the game represented by Algorithm~\ref{toyalgo} continues forever;
	more precisely, regardless of the coin flip results, $\mathcal{S}$ can construct a run of Algorithm~\ref{toyalgo}
	in which all the processes loop forever (Theorem~\ref{LinearizableIsWeak} in Section~\ref{ltoysection}).
We then show that if the registers are $\sly$, then all the correct processes\footnote{A process is correct
	if it takes infinitely many steps. We assume that processes continue to take steps (forever)
	even after returning from the algorithm in lines~\ref{halt0} or line~\ref{halt1}.}
	 return from the algorithm with probability~$1$ (Theorem~\ref{WSLinearizableIsStrong} in Section~\ref{wsltoysection}).

At high-level, the main idea of the proof is as follows.
Assume the register $R_1$ is not atomic,
	so each of its operations spans an interval of time, and operations on $R_1$ can be concurrent.
Consider the time $t$ after $p_0$ flipped the coin (line~\ref{pcoin1-1}).
Suppose at that time $t$, the write of $[1,j]$ into $R_1$ by $p_1$ is still pending and concurrent
	with the completed write of $[0,j]$ into $R_1$ by $p_0$.

If $R_1$ is linearizable,
	then adversary has the power to linearize the two writes in either order:
	$[0,j]$ before $[1,j]$, or $[1,j]$ before $[0,j]$.
So based on the outcome $c$ of the coin flip,
	the adversary can ensure that all the players
	read $[c,j]$ and then $[1-c,j]$ from $R_1$
	which forces them to stay in the game.
	
If, on the other hand, $R_1$ is write strongly-linearizable,
	the adversary does not have this power:
	at the time $p_0$ completes its write of $[0,j]$ into $R_1$
	(and therefore before the adversary can see the result of the coin flip)
	the adversary must decide whether the concurrent write of $[1,j]$ by $p_1$
	is linearized before $[0,j]$ or not.
With probability at least 1/2,
	the result of the coin flip will not ``match'' this decision. 
So in each round, with probability at least 1/2, the players
	will \emph{not} read $[c,j]$ and then $[1-c,j]$ from $R_1$
	and so they will exit the game.
(Note that if $R_1$ is atomic, operations are instantaneous, and so of course the adversary has no power to continue the game forever.)

In Algorithm~\ref{toyalgo} only register $R_1$ is unbounded, but we can easily make
	$R_1$ bounded (see Appendix~\ref{btoysection}).

\subsection{Linearizability does not ensure termination}\label{ltoysection}

\begin{theorem}\label{LinearizableIsWeak}
If registers $R_1$, $R_2$, and $C$ are only linearizable,
	a~strong adversary $\mathcal{S}$ can construct a run of Algorithm~\ref{toyalgo} where all the processes execute
	infinitely many rounds.
\end{theorem}

\begin{proof}
Assume the registers of Algorithm~\ref{toyalgo} are only linearizable but \emph{not} $\sly$.
A strong adversary $\mathcal{S}$ can construct an infinite run of Algorithm~\ref{toyalgo} as follows (Figure~\ref{toy}):

\begin{figure}[!htb]
   		 \centering 
    		\includegraphics[width=0.8\textwidth]{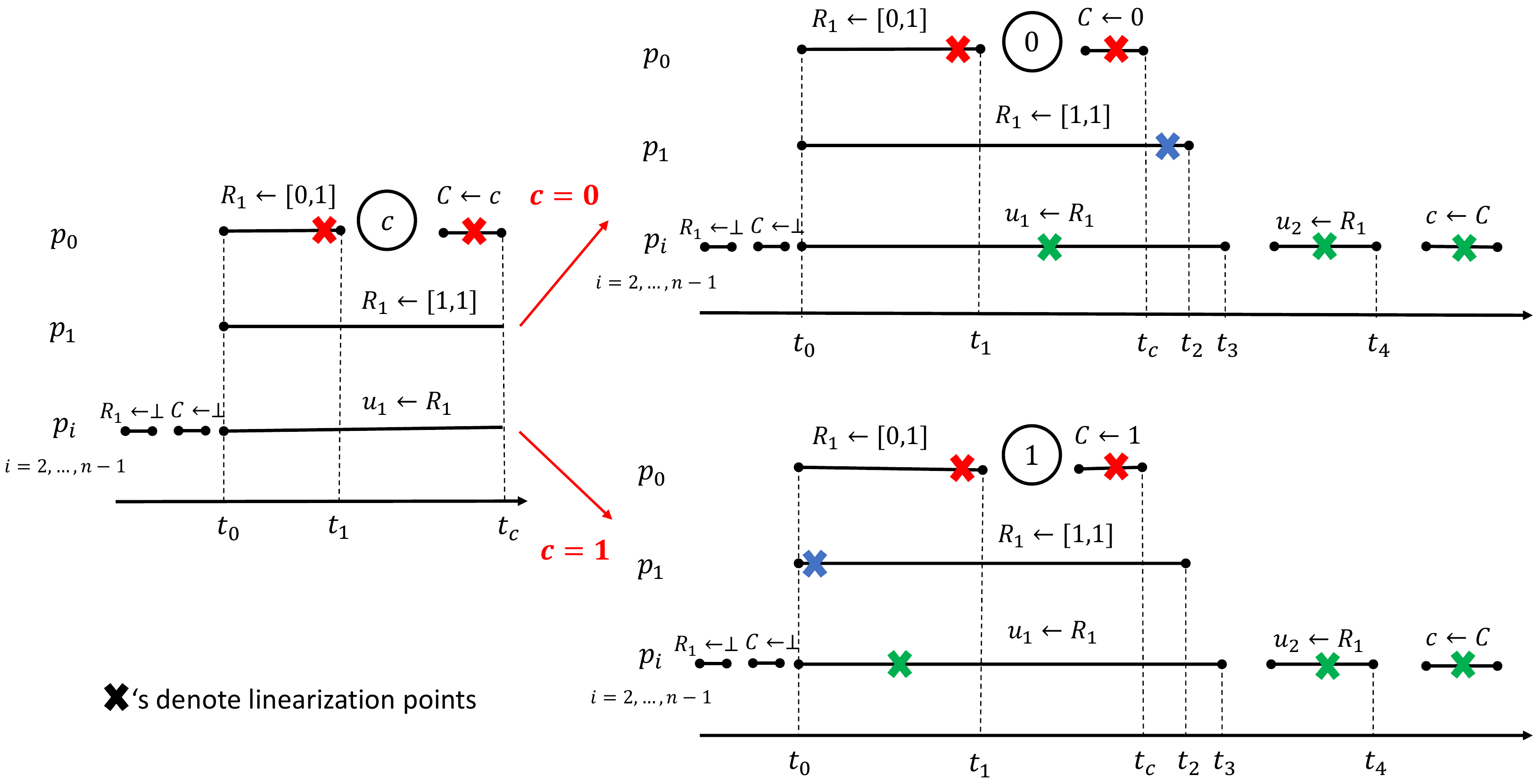}
   		 \caption{Phase 1 in round $j=1$ of an infinite execution} 
   		 \label{toy}
	\end{figure}

\medskip
\noindent
\textbf{Phase 1} (of round $j=1$):
\begin{enumerate}

\item Processes $\procs$ write $\bot$ into $R_1$ and $C$ in lines~\ref{cleanr1} and~\ref{cleanc}.

\item At some time $t_0$ after all the above write operations complete,
	process $p_0$ starts writing~$[0,1]$ into $R_1$ in line~\ref{pwrite1},
	process $p_1$ starts writing~$[1,1]$ into $R_1$ in line~\ref{pwrite1},
	and $\Procs$ start reading $R_1$ in line~\ref{u1}.

\item At time $t_1 > t_0$, process $p_0$ completes its writing of $[0,1]$ into $R_1$ in line~\ref{pwrite1}.

\item After time $t_1$,
	process $p_0$ flips a coin in line~\ref{pcoin1-1} and writes the result into the shared register $C$ in line~\ref{pcoin1-2}.
	Let $t_c > t_1$ be the time when $p_0$ completes this write of $C$.

Depending on the result of $p_0$'s coin flip (and therefore the content of $C$),
	the adversary $\mathcal{S}$ continues the run it is constructing in one of the following two ways:

\textbf{Case 1}: $C=0$ at time $t_c$.

The continuation of the run in this case is shown at the top right of Figure~\ref{toy}.

\begin{enumerate}

\item At time $t_2 > t_c$,
		$p_1$ completes its writing of [1,1] into $R_1$ in line~\ref{pwrite1}.
		
		Note that \emph{both} $p_0$ and $p_1$ have now completed Phase 1 of round $j=1$.

\item The adversary $\mathcal{S}$ linearizes the write of $[1,1]$ into $R_1$ by $p_1$
	\emph{after} the write of $[0,1]$ into $R_1$ by~$p_0$.

\item Note that $\procs$ are still reading $R_1$ in line~\ref{u1}.
Now the adversary linearizes these read operations \emph{between} the above write of $[0,1]$ by $p_0$ and the write of $[1,1]$ by~$p_1$.

\item At time $t_3 > t_2$,
		$\Procs$ complete their read of $R_1$ in line~\ref{u1}.
By the above linearization, they read $[0,1]$, and so they set (their local variable) $u_1 = [0,1]$ in line~\ref{u1}.
		
\item Then $\Procs$ start and complete their read of $R_1$ in line~\ref{u2}.
Since (1)~these reads start \emph{after} the time $t_2$ when $p_1$ completed its write of $[1,1]$ into $R_1$,
	and (2)~this write is linearized \emph{after} the write of $[0,1]$ by $p_0$ into $R_1$,
	$\Procs$ read~$[1,1]$.
So they all set (their local variable) $u_2 = [1,1]$ in line~\ref{u2}.
Let $t_4 > t_3$ be the time when every process in $\{\procs\}$ has set $u_2 = [1,1]$ in line~\ref{u2}.

\item After time $t_4$, $\Procs$ start reading $C$ in line~\ref{rcoin1}.
		Since $C =0$ at time~$t_c$ and it is not modified thereafter,
		$\procs$ read $0$ and set (their local variable) $c = 0$ in line~\ref{rcoin1}.

	So at this point, $\Procs$ have $u_1 = [0,1]$, $u_2 = [1,1]$ and $c = 0$.
	
\item Then $\procs$ execute line~\ref{guard0}, and find that
	the condition ($u_1 = \bot$ \textbf{or} $u_2 = \bot$ \textbf{or} $c = \bot$)
	of this line does not hold,
	and so they proceed to execute line~~\ref{guard1}.
	
\item When $\procs$ execute line~\ref{guard1}, they find that the condition ($u_1 \neq [c,j]$ \textbf{or}  $u_2 \neq [1-c,j]$)
	of this line does \emph{not} hold, because they have $u_1 = [c,1] = [0,1]$ and $u_2 = [1- c,1] = [1,1]$.

	   So $\procs$ complete Phase 1 of round $j=1$
	   without exiting in line~\ref{exit1}.
	Recall that both $p_0$ and $p_1$ also completed Phase 1 of round $j=1$ without exiting.

\end{enumerate}

\textbf{Case 2}: $C=1$ at time $t_c$.

The continuation of the run in this case is shown at the bottom right of Figure~\ref{toy}.
This continuation is symmetric to the one for Case 1:
	the key difference is that the adversary $\mathcal{S}$ now linearizes
	$p_1$'s write of $[1,1]$ into $R_1$ \emph{before} $p_0$'s write of $[0,1]$ into $R_1$,
	and so processes $\procs$ have $u_1 = [c,1]=[1,1]$ and $u_2 = [1-c,1]=[0,1]$ and so they will also complete Phase 1 without existing in line~\ref{exit1}.

Thus in both Case 1 and Case 2, all the $n$ processes complete Phase 1 of round $j=1$ without exiting,
	and are now poised to execute Phase 2 of this round.
The adversary $\mathcal{S}$ extends the run that it built so far as follows (Figure~\ref{toy2}).

\end{enumerate}

\begin{figure}[!htb]
   		 \centering 
    		\includegraphics[width=0.5\textwidth]{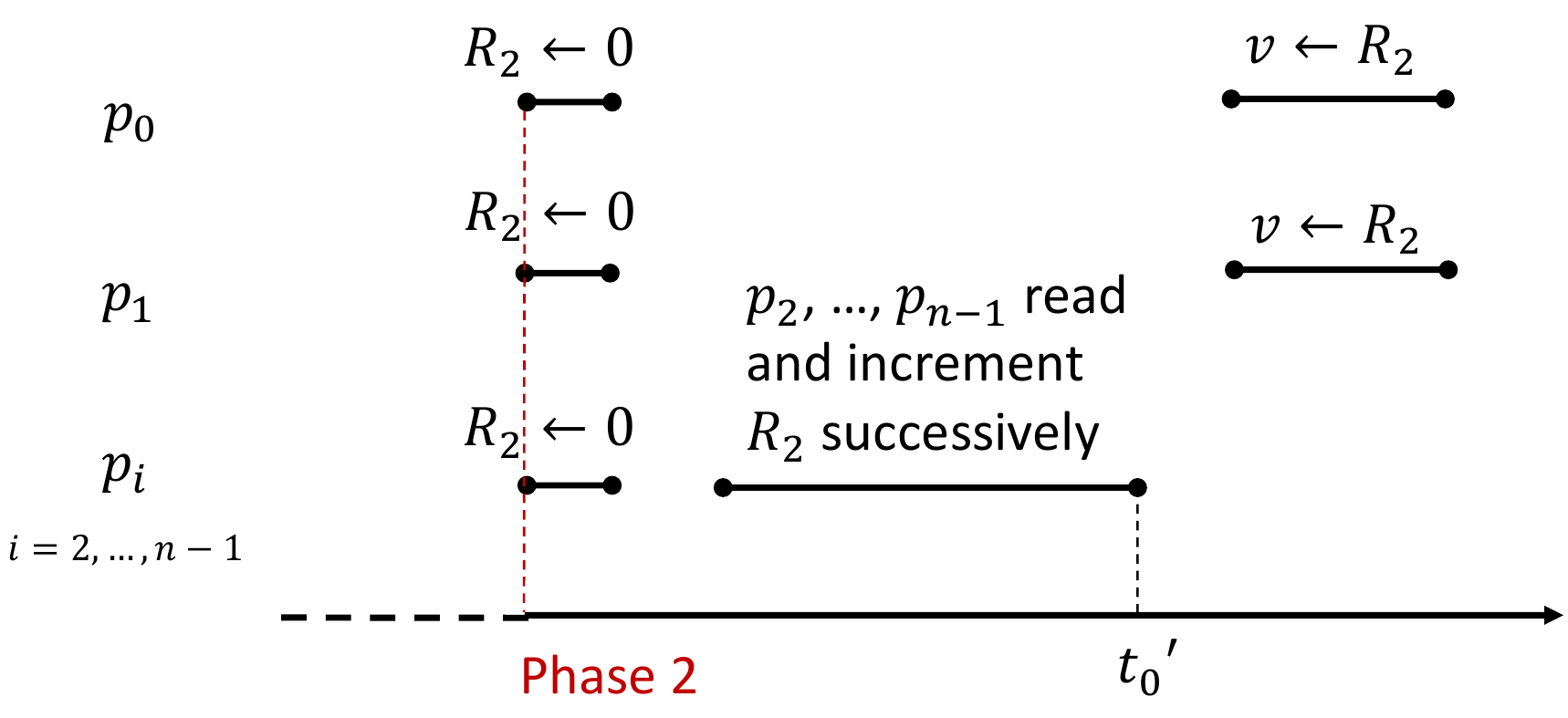}
   		 \caption{Phase 2 in round $j=1$ of an infinite execution} 
   		 \label{toy2}
\end{figure}

\noindent
\textbf{Phase 2} (of round $j=1$):

\begin{enumerate}

\item Processes $p_0$ and $p_1$ write $0$ into $R_2$ in line~\ref{cleanr2},
	and processes $\procs$ write $0$ into $R_2$ in line~\ref{cleanr2-2}.

\item After all the above write operations complete,
	$\Procs$ successively read and increment $R_2$
	by executing lines~\ref{pread2}--\ref{pwrite2}
	in the following order:
	$p_2$ executes lines~\ref{pread2}--\ref{pwrite2},
	and then, for  each  \mbox{$i\in \{3, \ldots, n-2\}$},
	process $p_{i+1}$ starts reading $R_2$ in line~\ref{pread2} \emph{after}
	$p_i$ completes its write of $R_2$ in line~\ref{pwrite2}.

Let $t'_0$ be the time when the above $n-2$ write operations by $\procs$ have completed.
	
	Note that at time $t'_0$: (i) register $R_2$ contains $n-2$, and
	(ii) all $\Procs$ have completed~Phase~2 of round $j=1$.
	
\item After time $t'_0$, $p_0$ and $p_1$ read $R_2$ into $v$ in line~\ref{v1},
	and so they set $v =  n-2$ in that line.

\item Then $p_0$ and $p_1$ execute line~\ref{guard2}
	and find that the condition ``$v <  n-2$'' of this line does \emph{not} hold.
         So $p_0$ and $p_1$ complete Phase 2 of round $j=1$
	 without exiting in line~\ref{exit2}.

	
Thus all the $n$ processes $p_0, p_1, \ldots, p_{n-1}$,
	have completed Phase 2 of round $1$ without exiting;
	they are now poised to execute round $j=2$.

\end{enumerate}

\noindent
The adversary $\mathcal{S}$ continues to build the run
	by repeating the above scheduling of $p_0, p_1, \ldots, p_{n-1}$ for rounds $j=2,3 \ldots$.
This gives a non-terminating run of Algorithm~\ref{toyalgo} with \mbox{probability~1:}
	in this run,
	all processes are correct, i.e., they take an infinite number of steps, but they loop forever
	and never reach the return statement in lines~\ref{halt0} or line~\ref{halt1}.
\end{proof}

\subsection{Write strong-linearizability ensures termination}\label{wsltoysection}


In Appendix~\ref{toyterminatewithwsl}, we prove:

\begin{restatable}{theorem}{toyterminate}\label{WSLinearizableIsStrong} 
If registers $R_1$, $R_2$, and $C$ are $\sly$,
	 then Algorithm~\ref{toyalgo} terminates with probability 1 against a strong adversary.
	 \MP{Expected 2 rounds?}
\end{restatable}

Combining Theorems~\ref{LinearizableIsWeak} and ~\ref{WSLinearizableIsStrong},
	we have:
\begin{corollary}\label{smallbob}
If $R_1$, $R_2$, and $C$ are
\begin{enumerate}
\item  only linearizable,
    then a strong adversary can prevent the termination~of~Algorithm~\ref{toyalgo};

\item $\sly$, 
    then Algorithm~\ref{toyalgo} terminates with probability 1 against a strong adversary.
 
\end{enumerate}
\end{corollary}

\rmv{
\begin{proof}
Consider any round $j \ge 1$.
By Lemma~\ref{exitchance},
	with probability at least $1/2$,
	no process enters round $j+1$.
Since this holds for every round $j \ge 1$,
	then, with probability 1,
	all the correct processes
	return in lines~\ref{halt0} or line~\ref{halt1}
	within a finite number of rounds $r$,
	and the expected value of $r$ is 2.
\end{proof}

}


Consider any randomized algorithm $\mathcal{A}$ that solves some task $T$,
	such as consensus, 
	for $n \ge 3$ processes $p_0, p_1, p_2, \ldots, p_{n-1}$,
	and terminates with probability 1 against a strong adversary.
Using $\mathcal{A}$, we can construct a corresponding randomized algorithm $\mathcal{A}'$ as follows:
	every process $p_i$ with $i\in\{0,1,2,...,n-1\}$
	first executes Algorithm~\ref{toyalgo};
	if~$p_i$ returns then it executes algorithm~$\mathcal{A}$.
From Corollary~\ref{smallbob} we have:

\begin{corollary}\label{bigbob}
Let $\mathcal{A}$ be any randomized algorithm that solves a task $T$ 
	for $n \ge 3$ processes and terminates with probability~1 against a strong adversary.
There is a corresponding randomized algorithm $\mathcal{A}'$
	that solves $T$ for $n \ge 3$ processes such that:

\begin{enumerate}

\item $\mathcal{A}'$ uses a set $\mathcal{R}$ of three shared registers
	in addition to the set of base objects of $\mathcal{A}$.

\item If the registers in $\mathcal{R}$ are only linearizable,
    then a strong adversary can prevent the termination~of~$\mathcal{A}'$.

\item If the registers in $\mathcal{R}$ are $\sly$, 
    then $\mathcal{A}'$ terminates with probability 1 against a strong adversary.\footnote{$\mathcal{A}'$ also terminates if the registers in $\mathcal{R}$ are atomic, because atomic registers are $\sly$.}
    
\end{enumerate}
\end{corollary}

\section{Implementing write strongly-linearizable MWMR registers from SWMR registers}\label{sectionwslalgo}

To implement a write strongly-linearizable MWMR register $\REG$,
	we must be able to linearize all the write operations ``on-line''
	without looking at what may happen in the future.
The challenge is that at the moment $t$ a write operation $w$ completes,
	for each write $w'$ that is concurrent with $w$ and is still pending at time $t$,
	we must have enough information to \emph{irrevocably} decide
	whether $w'$ should be linearized before or after $w$;
	of course this linearization order must be consistent with the values that
	processes previously read and will read in the future from $\REG$.
Using simple ``Lamport clocks'' to timestamp and linearize write operations does not seem to work:
	in the next section, we give an implementation showing that Lamport clocks
	are sufficient to implement a linerizable MWMR register,
	but this implementation is \emph{not} $\sly$.

In this section we give an implementation of a MWMR register
	from SWMR registers ($\Awsl$), and prove that it is write strongly-linearizable.
This is a modification of an implementation given in~\cite{swlamport}
	and it uses vector clocks to timestamp write operations.
The question is how to use vector timestamps to linearize write operations \emph{on-line}.
Specifically, at the moment $t$ a write operation $w$ completes,
	for each operation $w'$ that is concurrent with $w$ and still pending at time $t$, 
	how do we decide the order of $w'$ with respect to $w$?
Note that at time $t$, while the vector timestamp of $w$ is known,
	the vector timestamp of such $w'$ may not be known:
	it is still being computed (it may be incomplete with just a few entries set).
The proof of linearization given in~\cite{swlamport} does not work here:
	that proof can linearize all the write operations after seeing their \emph{complete} vector timestamps;
	and it can do so because linearization is done ``off-line''.

$\Awsl$ uses SWMR registers $Val[i]$ for $i=1,2,...,n$.
Each write operation $w$ is timestamped with a vector timestamp;
	roughly speaking, this represents the number of write operations that every
	process performed ``causally before'' $w$.
Each $Val[k]$ contains the latest value that $p_k$ wrote to $\REG$
	with its corresponding vector timestamp. 
To write a value $v$ into $\REG$,
	a process $p_k$ first constructs a new timestamp $\nts$, incrementally one component at a time,
	by successively reading $Val[1], \ldots, Val[n]$ (lines~\ref{timestampb}--\ref{timestampe});
	then $p_k$ writes the tuple $(v,\nts)$ into $Val[k]$ (line~\ref{write});
	finally $p_k$ resets its $\nts$ to $[ \infty, \ldots, \infty]$ (as we will see,
	this is important for the $\str$).
To read $\REG$,
	a process $p$ first reads all $Val[-]$ (lines~\ref{collecttsb}--\ref{collecttse});
	then $r$ returns the value $v$ with the greatest vector timestamp in \emph{lexicographic order $\le$}
	(lines~\ref{max}--\ref{readreturn}). 
Note that this is a total order.


We now prove that this MWMR implementation is indeed $\sly$.
Before we do so, we first illustrate the problem that we mentioned earlier,
	namely, how to linearize write operations on-line based on incomplete vector timestamps,
	and then we give some intuition on how this problem is solved.


\begin{algorithm}[!h]
\caption{Implementing a $\sly$ MWMR register $\REG$ from SWMR registers}
\begin{flushleft}
\textsc{Shared Object:}

For $i=1,2,...,n$:

$Val[i]$: SWMR register that contains a tuple $(v,ts)$ where $v$ is a value and $ts$ is a vector timestamp;

\hspace*{0.95cm}initialized to $( 0,[0\ldots0] )$ where $0$ is the initial value of $\REG$ and  $[0\ldots0]$ is an array of length $n$.

\vspace{0.2cm}

\textsc{Local Object:}

$new\_ts$: For each process, a local array of length $n$;
	initialized to $[\infty, \ldots, \infty]$.
\end{flushleft}
\vspace{0.05cm}
\begin{algorithmic}[1]
\STATEx
When writer $p_k$ writes $v$ to $\REG$:  \hfill // $1\le k \le n$

\FOR{$i = 1$ to $n$} \label{timestampb}
	\IF{$i\ne k$}
	\STATE \label{ts1} $\nts[i] \gets (Val[i].ts)[i]$ \hfill // $Val[i].ts$ is the vector timestamp of the tuple in $Val[i]$
		\ELSE
	\STATE \label{ts2} $\nts[i] \gets (Val[i].ts)[i]+1$
	\ENDIF

\ENDFOR\label{timestampe}

\STATE $Val[k] \gets (v,\nts)$ \hfill // write to shared register \label{write}
\STATE $ \nts \gets [ \infty, \ldots, \infty]$
\STATE \label{wreturn}\textbf{return} done
\newline

\STATEx When a process reads from $\REG$:

\FOR{$i = 1$ to $n$}\label{collecttsb}

	\STATE $(v_i,ts_i)\gets Val[i]$\label{readall} 
\ENDFOR\label{collecttse}

\STATE let $j$ be such that $ts_j = \max \{ts_1,...,ts_n\}$ \label{max}  \hfill // lexicographic max

\STATE \textbf{return} $(v_j,ts_j)$\label{readreturn} 
\end{algorithmic}\label{wsl}
\end{algorithm}

Consider a write operation $w_2$ that finishes at time $t$,
	as illustrated in Figure~\ref{diff1}.
For each write operation that is active at time $t$
	we must decide whether it should be linearized
	before or after $w_2$.
We cannot treat all such operations in the same way:
	we cannot simply linearize all of them before
	or all of them after $w_2$.
This is because the linearization order of these write operations
	depends on their timestamps
	(so as to respect the order of the read operations
	that read their values),
	which may not yet be fully formed at time $t$.
Indeed, the timestamp of a write operation active at $t$
	may end up being greater than, or smaller than,
	the timestamp of $w_2$.
For example, in Figure~\ref{diff1}
	the timestamps eventually computed
	by the write operations $w_1$ and $w_3$,
	which are active at time $t$,
	end up being, respectively,
	greater than and smaller than the timestamp of $w_2$.
As we will see, by looking carefully at the progress
	that each of $w_1$ and $w_3$ has made by time $t$
	towards computing its vector timestamp,
	we can determine, \emph{at time $t$},
	the correct linearization order of
	$w_1$ and $w_3$ relative to $w_2$.
We do so by
	(a) initializing the timestamp of each write to $[ \infty, \ldots, \infty]$ 
	(so it gets smaller and smaller while it is being formed); and
	(b) ordering the writes by comparing their (possibly incomplete) timestamps
	in lexicographical order.
This makes it possible to linearize the write operations on-line.

\begin{figure}[!h]
\centering
\includegraphics[width=0.7\textwidth]{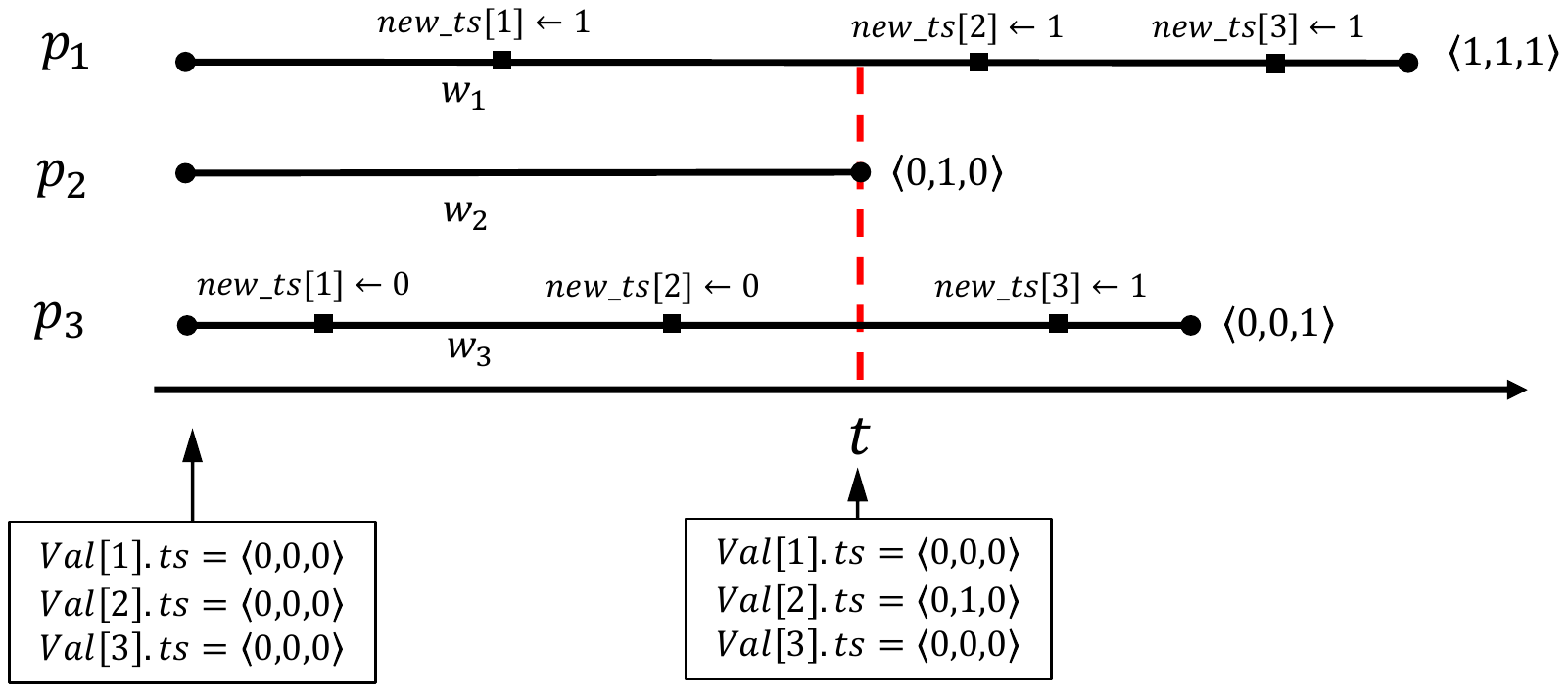}
  \caption{Three concurrent writes under $\Awsl$}
  \label{diff1}
\end{figure}

\begin{algorithm}[!htb]
\caption{A $\str$ function $f$ for the set of histories $\seth$ of $\Awsl$}
\begin{flushleft}
\textbf{Input:} a history $H \in \seth$

\textbf{Output:} $\seq{}$, a sequential history of the operations in $H$
\end{flushleft}
\vspace{0.1cm}
\begin{algorithmic}[1]
\STATEx linearization of write operations in $H$
\STATE \label{seq0}$\wseq{0} \leftarrow()$
\STATE $m \gets$ the number of operations that write to $Val[-]$ in $H$    \hfill//{ $m$ can be $\infty$}
\FOR{$i = 1, 2, \ldots,m$}\label{findleader1}

\STATE \label{findleader2} $t_i \gets$ the time of the $i$th write to a register $Val[-]$ in $H$ (line~\ref{write} of $\Awsl$)

\STATE \label{findleader3} $w_i \gets$ the operation that writes to $Val[-]$ at time $t_i$


\IF{$w_i \notin \wseq{i-1}$}\label{line13} \label{findleader4}

\STATE\label{C} $\sC{i} \leftarrow \{ w~|~w \text{ is a write operation such that } w \notin \wseq{i-1}  \text{ and } w \text{ is active at time } t_i \text{ in } H\}$

 
\STATE \label{newts} $\forall w\in \sC{i}$, 
	$\pts{w}{i}\gets \text{the value of $\nts$,
	 at time $t_i$}$,
	 of the process executing $w$

\STATE \label{Bt} $\B{i} \leftarrow \{w~|~w\in \sC{i} \text{ and } \pts{w}{i} \le \pts{w_i}{i}\}$
\STATE \label{seqi1} $\wseq{i} \gets \wseq{i-1}\circ (\text{the sequence of operations $w \in \B{i}$ in increasing order of $\pts{w}{i}$)}$

\ELSE
\STATE $\sC{i} \leftarrow \emptyset$; $\B{i} \leftarrow \emptyset$
\STATE \label{seqi2}$\wseq{i} \gets \wseq{i-1}$
\ENDIF
\ENDFOR

\IF{$m = \infty$}
\STATE \label{wseq1}$\wseq{} \gets \lim_{i\rightarrow \infty} \wseq{i}$
\ELSE
\STATE \label{wseq2}$\wseq{} \gets \wseq{m}$
\ENDIF

\smallskip

\STATEx linearization of read operations in $H$
\STATE $\seq{}' \gets \wseq{}$
\FOR{every value $(v,ts)$ that processes read in $H$}\label{fmapread}
\STATE \label{R} $\mathcal{R} \leftarrow \{ r~|~r \text{ is a read operation that returns $(v,ts)$ in $H$} \}$
\STATE \label{S} $\SR \leftarrow$ the sequence of operations in 
$\mathcal{R}$ in increasing order of their start time
\IF{$ts=[0,\ldots,0]$} \label{readinit1}
\STATE \label{readinit2} prepend $\SR$ to $\seq{}'$
\ELSE 
\STATE \label{followwrite1} $w \leftarrow$ the operation in $\seq{}'$ that writes $(v,ts)$
\STATE \label{followwrite2} insert $\SR$ after $w$ and before any subsequent write operation in $\seq{}'$ 
\ENDIF


\ENDFOR\label{endfmapread}
\STATE $\seq{} \gets \seq{}'$
\end{algorithmic}\label{f}
\end{algorithm}

 
To prove that the MWMR implementation given by $\Awsl$ is $\sly$,
	we give a $\str$ function $f$ for the set of histories $\seth$ of this algorithm.
We describe $f$ as an algorithm ($\f$)
	that takes as input any finite or infinite history $H \in \seth$ 
	and outputs a sequential history $\seq{}$ that satisfies properties (L) and (P) of Definition~\ref{defwsl}.
Intuitively, $\f$ linearizes all the write operations \emph{on-line}, as follows.
 It scans the input history $H$ by increasing time;
	 while doing so it maintains a sequence $\wseq{}$ of write operations that it has linearized so far.
When it sees that, at some time $t_i$, a write operation $w_i$ writes to $Val[-]$,
	it first checks whether $w_i$ was already linearized i.e., whether $w_i$ is in $\wseq{}$ (lines~\ref{findleader1}--\ref{findleader4}).
If $w_i$ is not in $\wseq{}$,
	it forms the set $\sC{i}$ of all the write operations that are ``active'' at time $t_i$
	and are not yet in $\wseq{}$ (line~\ref{C}).\footnote{
An operation that starts at time $s$ and completes at time $f$ is \emph{active} at time $t$ if $s \le t \le f$.
}
It then determines the (possibly incomplete) timestamp of each operation in $\sC{i}$ at time $t_i$ (line~\ref{newts});
	note that $w_i$ is in $\sC{i}$ and $w_i$'s timestamp, denoted $\pts{w_i}{i}$, is complete.
Finally, it selects the operations in $\sC{i}$ whose (possibly incomplete) timestamps are smaller than or equal to $\pts{w_i}{i}$ (line~\ref{Bt}),
	and then it appends them to the linearization sequence $\wseq{}$ in increasing timestamp order (line~\ref{seqi1}).
To linearize the read operations, it collects all the read operations that return some value $(v,ts)$,
 	and linearizes them after the write operation that writes $(v,ts)$, in increasing start time order (lines~\ref{fmapread}--\ref{endfmapread}).

Intuitively,
	$\f$ gives a $\sly$ function $f$ because:
	 (1) by scanning $H$ in increasing time, it linearizes each write operation by the time the operation completes without ``peeking into the future'',
	and (2) it only \emph{appends} write operations to $\wseq{}$, and so it satisfies the ``prefix property'' (P) of Definition~\ref{defwsl}.
In Appendix~\ref{proofwsl},
	we prove that $\f$ defines a $\str$ function for $\seth$ and
	thus show:


\begin{restatable}{theorem}{algowsl}\label{aiswsl}
$\Awsl$ is a $\sly$ implementation of a MWMR register from SWMR registers.
\end{restatable}

Helmi \emph{et. al} show that there is no strongly linearizable implementation of MWMR registers from SWMR registers
	(Corollary 3.7 in \cite{sl12}).
Thus:
\begin{corollary}
$\Awsl$ is not a strongly linearizable implementation of a MWMR register.
\end{corollary}	

This implies that strong linearizability is strictly stronger than write strong-linearizability.

\section{Achieving write strong-linearizability is harder than achieving linearizability}\label{sectionlamport}	

As we will see, every linearizable implementation of an SWMR register
	is necessarily write strongly-linearizable.
In contrast, here we prove that there is 
	a linearizable implementation of an MWMR register from SWMR registers (namely, Algorithm~\ref{Al})
	that is \emph{not} write strongly-linearizable.

	 
\begin{algorithm}[!b]
\caption{Implementing a linearizable MWMR register $\REG$ from SWMR registers}
\begin{flushleft}
\textsc{Shared Object:}

For $i=1,2,...,n$:

$Val[i]$: SWMR register that contains a tuple $(v,ts)$ 
	where $v$ is a value and $ts$ is a tuple of the form~$\langle sq, pid\rangle$;

\hspace*{1.2cm}initialized to $(0,\langle 0,i \rangle)$.

\vspace{0.2cm}

\textsc{Local Object:}

$\nsq$: a register initialized to $0$.

$\nts$: a register initialized to $\langle 0,k \rangle$ for process $p_k$.
\end{flushleft}
\vspace{0.05cm}

\begin{algorithmic}[1]

\STATEx When writer $p_k$ writes $v$ to $\REG$ \hfill//$1 \le k \le n$
\FOR{$i = 1$ to $n$} \label{wcollectts1}
	\STATE $(v_i,ts_i) \gets Val[i]$ \label{wcollectts2}
\ENDFOR\label{wcollectts3}
\STATE $\nsq \gets \max \{ts_1.sq,...,ts_n.sq\} +1$ \label{nsq}
\STATE $\nts \gets \langle new\_sq, k \rangle$ \label{nts}
\STATE $Val[k] \gets (v,new\_ts)$\label{writeAl}
\STATE \textbf{return} done
\newline 

\STATEx When a process reads from $\REG$ 
\FOR{$i = 1$ to $n$}\label{collectts1}
	\STATE $(v_i,ts_i) \gets Val[i]$\label{collectts2}
\ENDFOR\label{collectts3}
\STATE let $j$ be such that $ts_j = \max \{ts_1,...,ts_n\}$ \label{tsmax} \hfill// lexicographic max
\STATE \textbf{return} $(v_j,ts_j)$ \label{returnmax}
\end{algorithmic}\label{Al}
\end{algorithm}

Algorithm~\ref{Al} implements a MWMR register $\REG$ from SWMR registers $Val[i]$ for $i=1,2,...,n$.
Each value written to $\REG$ is timestamped with tuple $\langle sq,pid\rangle$ 
	where $sq$ is a sequence number and $pid$ is the id of the process that writes the value;
	intuitively, these are Lamport clocks that respect the causal order of write events.
Each register $Val[k]$ contains the latest value that $p_k$ wrote to $\REG$
	with its corresponding timestamp. 
To write a value $v$ into $\REG$,
	$p_k$ first reads every register $Val[-]$ (lines~\ref{wcollectts1}--\ref{wcollectts3});
	then $p_k$ forms a new sequence number $\nsq$ by 
	incrementing the maximum sequence number 
	that it read from $Val[-]$ (line~\ref{nsq});
	finally, $p_k$ writes the tuple $(v,\nts)$, where  $\nts = \langle \nsq, k\rangle$, into $Val[k]$ (lines~\ref{nts}--\ref{writeAl}).
To read $\REG$,
	a process $p$ first reads all registers $Val[-]$ (lines~\ref{collectts1}--\ref{collectts3});
	then $r$ returns the value $v$ with the greatest timestamp in lexicographic order
	(lines~\ref{tsmax}--\ref{returnmax}).

Intuitively, Algorithm~\ref{Al} implements a linearizable MWMR register:
	the write operations can be linearized by their timestamps (which form a total order);
	the read operations are linearized according to the value that they read.
The proof of the following theorem is given in Appendix~\ref{lamport1}.

\begin{restatable}{theorem}{lampol}\label{lamportisl}
$\Al$ is a linearizable implementation of a MWMR register from SWMR registers.
\end{restatable}


The implemented register, however, is not $\sly$: roughly speaking the information provided by Lamport clocks
	is not sufficient to linearize the write operations on-line.

\begin{theorem}
$\Al$ is not a $\sly$ implementation of a MWMR register.
\end{theorem}

\begin{proof}
Consider the set of histories $\seth$ of $\Al$ executed by $n=3$ processes,
	 $p_1$, $p_2$ and $p_3$.
To prove that $\Al$ is not a $\sly$ implementation, 
		we show that $\seth$ is not \linebreak $\sly$. 
More precisely,
	we prove that for any function $f$ that maps histories in $\seth$ to sequential histories,
	there exist histories $G,H\in \seth$ such that $G$ is a prefix of $H$ but $f(G)$ is not a prefix of~$f(H)$.
	

\begin{figure}[!b]
\centering
\includegraphics[width=1\textwidth]{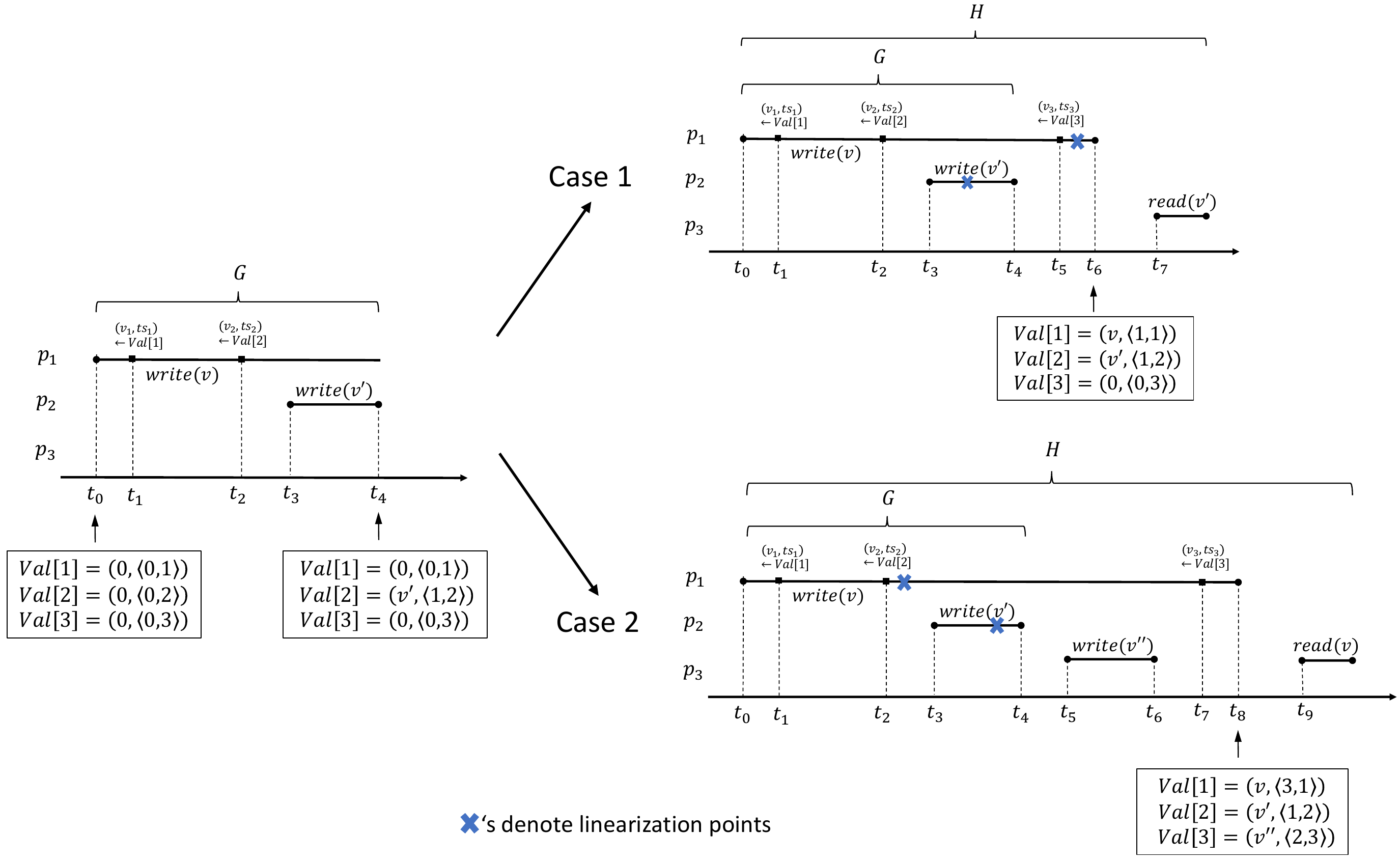}
  \caption{Histories $G$ and $H$ for Cases 1 and 2}
  \label{Hg}
\end{figure}

Let $f$ be a function that maps histories in $\seth$ to sequential histories.
Consider the following history $G \in \seth$ (shown at the left of Figure~\ref{Hg}): \XMP{caption of fig??}
\begin{itemize}
\item Initially, $\REG$ contains $0$, and each register $Val[i]$ contains $(0,\langle 0,i\rangle)$.
\item At time $t_0$,
	process $p_1$ starts an operation $\W{1}$ to write a value $v \ne 0$ to $\REG$.
By lines \ref{wcollectts1}--\ref{wcollectts2},
	$p_1$ first reads $(0,\langle 0,1\rangle)$ from $Val[1]$ into $(-,ts_1)$, 
	say at time $t_1$,
	and then reads $(0,\langle 0,2\rangle)$ from $Val[2]$ into $(-,ts_2)$, 
	say at time $t_2$.
Thus,
	$p_1$ now has $ts_1 = \langle 0,1\rangle$ 
	and $ts_2 = \langle 0,2\rangle$,
	i.e., $ts_1.sq = ts_2.sq = 0$.
\item At time $t_3 > t_2$,
	process $p_2$ starts an operation $\W{2}$ to write a value $v'$ to $\REG$ such that $v' \ne v$ and $v' \ne 0$.
By lines \ref{wcollectts1}--\ref{wcollectts2},
		$p_2$ reads $(0, \langle 0,1\rangle)$ from $Val[1]$ into $(-,ts_1)$,
		 $(0, \langle 0,2\rangle)$ from $Val[2]$ into $(-,ts_2)$,
		 and $(0, \langle 0,3\rangle)$ from $Val[3]$ into $(-,ts_3)$.
Thus,
	$p_2$ now has $ts_1 = \langle 0,1\rangle$, 
	$ts_2 = \langle 0,2\rangle$,
	and $ts_2 = \langle 0,3\rangle$,
	i.e., $ts_1.sq = ts_2.sq = ts_3.sq = 0$.
By lines~\ref{nts}--\ref{writeAl},
	 $p_2$ then writes $(v',\langle 1,2\rangle)$ to $Val[2]$
	 and completes $\W{2}$ at time $t_4$.	
Then at time $t_4$,
	$Val[1]$ contains $(0,\langle 0,1\rangle)$, 
	 $Val[2]$ contains $(v',\langle 1,2\rangle)$, and
	$Val[3]$ contains $(0,\langle 0,3\rangle)$.
\end{itemize}


Since the write operation $\W{2}$ completes in $G \in \seth$ and $f$ is a linearization function for $\seth$,
	by property~\ref{p1} of Definition~\ref{defl},
	$\W{2}$ is in $f(G)$.
Since the write operation $\W{1}$ is concurrent with $\W{2}$,
there are two cases: (1)~$\W{1}$ is not before $\W{2}$ in $f(G)$,
	 or (2) $\W{1}$ is before $\W{2}$ in $f(G)$.

\textbf{Case 1}: $\W{1}$ is not before $\W{2}$ in $f(G)$.
Consider the following history $H \in \seth$ (shown at the top right of Figure~\ref{Hg}):
\begin{itemize}
\item $H$ is an extension of $G$, i.e., $G$ is a prefix of $H$.
\item At time $t_5> t_4$,
	$p_1$ continues the operation $\W{1}$ and reads $(0, \langle 0,3\rangle)$ from $Val[3]$ into $(-,ts_3)$ so it has 
	 $ts_3.sq = 0$.
By lines~\ref{nts}--\ref{writeAl},
	$p_1$ writes $(v,\langle 1,1\rangle)$ to $Val[1]$.
After that,
	$p_1$ completes $\W{1}$, say at time $t_6$.
At time $t_6$,
	$Val[1]$ contains $(v,\langle 1,1\rangle)$,
	$Val[2]$ contains $(v',\langle 1,2\rangle)$,
	and $Val[3]$ contains $(0,\langle 0,3\rangle)$.
\item At time $t_7> t_6$,
	$p_3$ starts a read operation $\RD{}$ to read $\REG$.
In lines~\ref{collectts1}--\ref{collectts2},
	 $p_3$ reads $ts_1 = \langle 1,1\rangle$, 
	$ts_2 = \langle 1,2\rangle$,
	and $ts_3 = \langle 0,3\rangle$.
Since $ts_2 =\langle 1,2\rangle > ts_1 =\langle 1,1\rangle > ts_3 =\langle 0,3\rangle$,
	by line~\ref{tsmax},
	$\RD{}$ returns $(v',\langle 1,2\rangle)$ -- the value written by $\W{2}$.
\end{itemize}
Since the read operation $\RD{}$ returns the value written by $\W{2}$ in $H\in \seth$, 
	and $f$ is a linearization function for~$\seth$, 
		by property~\ref{p3} of Definition~\ref{defl},
	 	$\RD{}$ is after $\W{2}$ and before any subsequent write operation in $f(H)$.
Since $\RD{}$ starts after $\W{1}$ completes,
		by property~\ref{p2} of Definition~\ref{defl},
	$\RD{}$ is after $\W{1}$ in $f(H)$.
Thus,
	$\W{1}$ is before $\W{2}$ in $f(H)$.
Since, by assumption,
	 $\W{2}$ is in $f(G)$ and $\W{1}$ is not before $\W{2}$ in $f(G)$,
	$f(G)$ is not a prefix of $f(H)$.


\textbf{Case 2}: $\W{1}$ is before $\W{2}$ in $f(G)$.
Consider the following history $H \in \seth$ (shown at the bottom right of Figure~\ref{Hg}):
\begin{itemize}
\item $H$ is an extension of $G$.
\item At time $t_5> t_4$,
	process $p_3$ starts an operation $\W{3}$ to write $v''$ to $\REG$ such $v'' \ne v$.
In lines~\ref{wcollectts1}--\ref{wcollectts2},
	$p_3$ reads $(0, \langle 0,1\rangle)$ from $Val[1]$ into $(-,ts_1)$,
	$(v', \langle 1,2\rangle)$ from $Val[2]$ into $(-,ts_2)$,
	and $(0, \langle 0,3\rangle)$ from $Val[3]$ into $(-,ts_3)$
Thus,
	$p_3$ now has $ts_1 = \langle 0,1\rangle$, 
	$ts_2 = \langle 1,2\rangle$,
	and $ts_2 = \langle 0,3\rangle$,
	i.e., $ts_1.sq = 0$, $ts_2.sq =1$, and $ts_3.sq = 0$.
By lines~\ref{nts}--\ref{writeAl},
	$p_3$ then writes $(v'',\langle 2,3\rangle)$ to $Val[3]$.
	and completes $\W{3}$ at time $t_6$.	
\item At time $t_7> t_6$,
	$p_1$ continues the operation $\W{1}$ and reads $(v'', \langle 2,3\rangle)$ from $Val[3]$ into $(-,ts_3)$ so it has 
	 $ts_3.sq = 2$.
By lines~\ref{nts}--\ref{writeAl},
	$p_1$ writes $(v,\langle 3,1\rangle)$ to $Val[1]$.
After that,
	$p_1$ completes $\W{1}$, say at time $t_8$.
At time $t_8$,
	$Val[1]$ contains $(v,\langle 3,1 \rangle)$,
	$Val[2]$ contains $(v',\langle 1,2 \rangle)$,
	and $Val[3]$ contains $(v'',\langle 2,3 \rangle)$.
\item At time $t_9 > t_8$,
	$p_3$ starts a read operation $\RD{}$ to read $\REG$.
In lines~\ref{collectts1}--\ref{collectts2},
	$p_3$ reads $(v, \langle 3,1\rangle)$ from $Val[1]$ into $(-,ts_1)$,
	$(v', \langle 1,2\rangle)$ from $Val[2]$ into $(-,ts_2)$,
	and $(v'', \langle 2,3\rangle)$ from $Val[3]$ into $(-,ts_3)$.
Since $ts_1 = \langle 3,1\rangle > ts_3 =\langle  2,3\rangle > ts_2 = \langle 1,2\rangle$,
	by line~\ref{tsmax},
	$\RD{}$ returns $(v,\langle 3,1 \rangle)$ -- the value written by $\W{1}$.
\end{itemize}

Since the read operation $\RD{}$ returns the value written by $\W{1}$ in $H \in \seth$, 
	and $f$ is a linearization function for $\seth$,
	by property~\ref{p3} of Definition~\ref{defl},
	 	$\RD{}$ is after $\W{1}$ and before any subsequent write operation in $f(H)$.
Since $\RD{}$ starts after $\W{2}$ completes,
		by property~\ref{p2} of Definition~\ref{defl},
	$\RD{}$ is after $\W{2}$ in $f(H)$.
Thus,
	$\W{2}$ is before $\W{1}$ in $f(H)$.
Since, by assumption,
	 $\W{1}$ is before $\W{2}$ in $f(G)$,
	$f(G)$ is not a prefix of $f(H)$.

Then, in both case 1 and case 2, 
	there is a history $H \in \seth$
	such that $G$ is a prefix of $H$ but $f(G)$ is not a prefix of $f(H)$.
Therefore
	the theorem holds.
\end{proof}

\section{Linearizable SWMR registers are necessarily write strongly-linearizable}\label{swsection}
	
In Appendix~\ref{proofsw}, we show that \emph{any} linearizable implementation of SWMR registers is $\sly$
	(this holds for message-passing, shared-memory, and hybrid systems).
Thus, the well-known ABD implementation of SWMR registers in message-passing systems
	is not only linearizable; it is actually write strongly-linearizable.


\begin{restatable}{theorem}{swe}\label{swequal}
Any linearizable implementation of a SWMR register is necessarily $\sly$.
\end{restatable}

\section{Concluding remarks}
As we mentioned earlier,
	many randomized algorithms in the literature use atomic MWMR registers.
An interesting open problem is to determine which ones
	can also (be made to) work in systems with atomic SWMR registers.
If we replace a randomized algorithm's atomic MWMR registers with
	\emph{linearizable}\linebreak
	implementations of MWMR registers from atomic SWMR registers,
	we may break it:
	it may lose some of its properties~\cite{sl11} including,
	as we showed in this paper,
	termination.
On the other hand, we cannot replace
	an algorithm's atomic MWMR registers with
	\emph{strongly linearizable}
	implementations of MWMR registers from atomic SWMR registers
	--- which would automatically preserve the correctness of this algorithm ---
	simply because no such implementation exists~\cite{sl12}.
But perhaps many (maybe even most) of the known randomized algorithms that use atomic MWMR registers
	do not actually need the full strength of strongly linearizable registers:
	in particular, they may not need the ``strong linearizability property'' to hold for
	\emph{both} read and write operations.\footnote{In fact, a recent paper
	shows that a well-known randomized algorithm by Aspness and Herlihy that assumes atomic registers
	does not actually need any linearizability property:
	it works correctly with registers that are only \emph{regular}~\cite{vassos20}.}
Randomized algorithms that need the strong linearizability properties only for \emph{write operations},
	also work in systems with atomic SWMR registers:
	just replace their MWMR registers with the register implementation given by $\Awsl$.

\section*{Acknowledgments}
We thank Wojciech Golab and Philipp Woelfel for their helpful comments. 

\bibliographystyle{abbrv}
\bibliography{sl}

\begin{appendices}


\section{Proof of Theorem~\ref{WSLinearizableIsStrong}}\label{toyterminatewithwsl}

\noindent
We now prove Theorem~\ref{WSLinearizableIsStrong} in Section~\ref{wsltoysection}, namely, if registers $R_1$, $R_2$, and $C$ are $\sly$,
	Algorithm~\ref{toyalgo} terminates with \mbox{probability~1} even against a strong adversary.
To prove this, we first show four lemmas (Lemmas~\ref{prev-reads}--\ref{C-non-bot}) about some safety properties of Algorithm~\ref{toyalgo}.
Specifically, Lemma~\ref{bigarestuck} and~\ref{smallarestuck}
	state that processes $p_0$ and $p_1$ on one side,
	and processes $\procs$ on the other side,
	remain within one round of each other.
Lemma~\ref{prev-reads} and~\ref{C-non-bot}
	state that the non-$\bot$ values that $\procs$ read from registers $R_1$ and $C$ in any round $j$
	were also written in~round~$j$.

To prove that Lemmas~\ref{prev-reads}-\ref{C-non-bot} hold,
	we rely on our assumption that registers $R_1$, $R_2$, and $C$ are \emph{linearizable}~\cite{linearizability}.
So in the following proofs,
	we refer to the linearization times of the read and write operations that are applied on these registers.
For example, when we say ``a process $p$ writes a value $v$ into $R_1$ \emph{at time $t$}'',
	we mean that this write of $R_1$ ``took effect'' at time $t$,
	in other words, this write operation is \emph{linearized at time $t$}
	(where $t$ is within the time interval of the operation).

In the following lemmas, we say that a process ``enters round $r \ge 1$'' if it executes line~\ref{enter1} or \ref{enter2} with $j =r$.

\begin{restatable}{lemma}{safetyone}\label{prev-reads}
For all $j\ge1$,
	if $p_i \in \{ \procs \}$ $\exs$ line~\ref{cleanr2-2} in round $j$,
	then
	$p_i$ previously read both $[b,j]$ and $[1-b,j]$ for some
	$b \in \{0, 1\}$ from register $R_1$ in lines~\ref{u1} and~\ref{u2} in round $j$.
\end{restatable}

\begin{proof}
Suppose $p_i \in \{ \procs \}$ $\exs$ line~\ref{cleanr2-2} in a round $j \ge1$.
Since $p_i$ $\exs$ line~\ref{cleanr2-2} in round $j$, $p_i$ did not exit in lines~\ref{exit0} and~\ref{exit1} in round $j$.
So $p_i$ found the condition of lines~\ref{guard0} and~\ref{guard1} to be false in round $j$.
Thus, $p_i$ found that $c \neq \bot$ and ($u_1 = [c,j]$ and $u_2 = [1-c,j]$) in lines~\ref{guard0} and~\ref{guard1} in round $j$.
Note that: (i) $u_1$~and~$u_2$ contain the values that $p_i$ read
	from register $R_1$ in lines~\ref{u1} and~\ref{u2} in round $j$, and
	(ii)~$c$~contains the value that $p_i$ read
	from register $C$ in lines~\ref{rcoin1} in round $j$; since $c \neq \bot$
	and the only non-$\bot$ values written in $C$ are $0$ or $1$, $c = b \in \{0,1\}$.
Thus, $p_i$ read $[b,j]$ and $[1-b,j]$ for
	$b \in \{0, 1\}$ from register $R_1$ in lines~\ref{u1} and~\ref{u2} in round $j$.
\end{proof}

\begin{restatable}{lemma}{safetytwo}\label{bigarestuck}
For all $j\ge1$,
	if $p_i \in \{ \procs \}$ $\exs$ line~\ref{cleanr2-2} in round $j$, then $p_0$ and $p_1$ previously entered round $j$.
\end{restatable}

\begin{proof}
Suppose $p_i \in \{ \procs \}$ $\exs$ line~\ref{cleanr2-2} in round $j$.
By Lemma~\ref{prev-reads}, $p_i$ previously read both $[b,j]$ and $[1-b,j]$ for some
	$b \in \{0, 1\}$ from register $R_1$ in lines~\ref{u1} and~\ref{u2} in round $j$.
So both $[0,j]$ and $[1,j]$ were previously written into $R_1$.
Since $p_0$ and $p_1$ are the only processes that write $[0,j]$ and $[1,j]$,
	and they only do so in (line~\ref{pwrite1} of) round $j$,
	we conclude that $p_0$ and $p_1$ previously entered round $j$.
\end{proof}
\vspace{5pt}

\begin{restatable}{lemma}{safetythree}\label{smallarestuck}
For all $j\ge1$,
	if $p_0$ or $p_1$ enters round $j+1$, then every $p_i \in \{ \procs \}$ previously reached line 37 in round $j$.
\end{restatable}

\begin{proof}
To prove this part, we show the slightly stronger claim that
	if $p_0$ or $p_1$ enters round $j+1$ at some time $t$,
	then every $p_i \in \{ \procs \}$ \emph{writes} $R_2$ in line~\ref{pwrite2} in round~$j$ before time $t$.
The proof is by induction on $j$.

\smallskip\noindent
$\bullet$ \textbf{\scshape{base case}:} Let $j=1$.
Suppose that a process $p \in \{p_0 , p_1 \}$ enters round $j+1=2$ at some time $t$.
We must show that $p_i$ writes $R_2$ in line~\ref{pwrite2} in round~$j=1$ before time $t$.

\begin{claim}\label{nc1}
Up to and including time $t$, no process in $\{ \procs \}$ writes $R_2$ in line~\ref{pwrite2} in any round~$r \ge~2$.
\end{claim}

\begin{proof}
Suppose, for contradiction, that some process $p' \in \{ \procs \}$ writes $R_2$
	in line~\ref{pwrite2} in a round~$r \ge 2$ at some time $\hat{t} \le t$.
So $p'$ $\exs$ line~\ref{cleanr2-2} in round $2$ at some time $\hat{t'} <\hat{t} \le t$, i.e., $p'$ $\exs$ line~\ref{cleanr2-2} in round $2$ \emph{before} $p \in \{p_0 , p_1 \}$ enters round $2$ --- a contradiction to Lemma~\ref{bigarestuck}.
\end{proof}

Note that \emph{before} $p$ enters round~$2$ at time $t$,
	$p$ does the following in round $1$:
	it writes $0$ into $R_2$ in line~\ref{cleanr2} at some time $t_1$,
	it reads $R_2$ into $v$ in line~\ref{v1} at some time $t_2$,
	and then it finds that $v \ge n-2$ in line~\ref{guard2} (because $p$ does \emph{not} exit in line~\ref{exit2});
	clearly, $t_1 <t_2 <t$.
Thus, some process $p_j \in \{ \procs \}$ must have written a value $u \ge n-2$ into $R_2$ in line~\ref{pwrite2}
	at some time $t' \in [t_1 , t_2]$ in some round~$r$.
Since $t' \le t_2 <t$, by Claim~\ref{nc1}, $r \le 1$.
So $p_j$ writes $u \ge n-2$ into $R_2$ in line~\ref{pwrite2} in round $1$
	at time $t' \in [t_1 , t_2]$~(*).

\begin{claim}\label{needaproof1}
If a process in $\{ \procs \}$ writes $k\ge1$ into $R_2$ in line~\ref{pwrite2} in round $1$ at some time $\hat{t} \le t_2$,
then at least $k$ distinct processes in $\{ \procs \}$ write $R_2$ in line~\ref{pwrite2} in round $1$ by time $\hat{t}$.
\end{claim}

\begin{proof}
The proof is by induction on $k$.

\begin{itemize}
\item \textsc{Base Case:} For $k=1$ the claim trivially holds.

\item \textsc{Induction Step:} Let $k \ge1$.
Induction Hypothesis ($\mathsection$):
	if a process in $\{ \procs \}$ writes $k$ into $R_2$ in line~\ref{pwrite2} in round $1$
	at some time $\hat{t} \le t_2$,
	then at least $k$ distinct processes in $\{ \procs \}$ write $R_2$ in line~\ref{pwrite2} in round $1$ by time $\hat{t}$.

Suppose that a process $p_{\ell} \in \{ \procs \}$ writes $k+1$ into $R_2$ in line~\ref{pwrite2} in round $1$ at some time $\twkp \le t_2$,
we must show that at least $k+1$ distinct processes in $\{ \procs \}$ write $R_2$ in line~\ref{pwrite2} in round $1$ by time $\twkp$.

Since $p_{\ell}$ writes $k+1$ into $R_2$ in line~\ref{pwrite2} in round $1$ at time $\twkp \le t_2$:

\begin{enumerate}

\item $p_{\ell} $ writes $0$ into $R_2$ in line~\ref{cleanr2-2} in round $1$ at some time $\two$, and

\item $p_{\ell}$ reads $k$ from $R_2$ in line~\ref{pread2} in round $1$ at some time $\trk$,
	such that $\two < \trk < \twkp \le t_2 < t$.

\end{enumerate}

\noindent
Thus, some process $p^* \in \{ \procs \}$ must have written $k$ into $R_2$
	in line~\ref{pwrite2} at some time $t^* \in [\two, \trk]$ in some round $r$.
Since $t^* < t$, by Claim~\ref{nc1}, $r=1$.
So $p^*$ writes $k$ into $R_2$ in line~\ref{pwrite2} in round $1$ at time $t^* \in [\two, \trk]$.
Since $t^* \le \trk <t_2$, by the Induction Hypothesis ($\mathsection$),
	at least $k$ distinct processes in $\{ \procs \}$ write $R_2$ in line~\ref{pwrite2} in round $1$ by time $t^*$.
Recall that $p_{\ell} \in \{ \procs \}$ also writes $R_2$ in line~\ref{pwrite2} in round $1$ at time $\twkp$.
Thus, since each process in $\{ \procs \}$ writes $R_2$ in line~\ref{pwrite2} in round $1$ \emph{at most once},
	at least $k+1$ distinct processes in $\{ \procs \}$ write $R_2$ in line~\ref{pwrite2}
	in round $1$ by time $\twkp$.
\end{itemize}
\end{proof}

\noindent
Recall that by (*), $p_j$ writes $u \ge n-2$ into $R_2$ in line~\ref{pwrite2} in round $1$
	at time $t' \in [t_1 , t_2]$.
Since $t' \le t_2$, by Claim~\ref{needaproof1},
	at least $n-2$ distinct processes in $\{ \procs \}$ write $R_2$ in line~\ref{pwrite2} in round $1$ by time~$t'$.
Thus \emph{every} process in $\{ \procs \}$ writes $R_2$ in line~\ref{pwrite2} in round $1$ by time~$t'$.
Since $t'\le t_2 < t$, every $p_i \in \{ \procs \}$ writes $R_2$ in line~\ref{pwrite2} in round $1$ before time $t$.

\medskip\noindent
$\bullet$ \textbf{\textsc{Induction Step:}} Let $j \ge1$.
Induction Hypothesis ($\star$):
	for all $r$, $1\le r \le j$,
	if $p_0$ or $p_1$ enters round $r+1$ at some time $t$,
	then $p_i$ writes $R_2$ in line~\ref{pwrite2} in round~$r$ before time $t$.

\noindent
Suppose that a process $p \in \{p_0 , p_1 \}$ enters round $j+2$ at some time $t$.
We must show that  $p_i$ writes $R_2$ in line~\ref{pwrite2} in round~$j+1$ before time $t$.

\begin{claim}\label{nc}
Up to and including time $t$, no process in $\{ \procs \}$ writes $R_2$ in line~\ref{pwrite2} in any round~$r \ge j+2$.
\end{claim}

\begin{proof}
Suppose, for contradiction, that some process $p' \in \{ \procs \}$ writes $R_2$
	in line~\ref{pwrite2} in a round~$r \ge j+2$ at some time $\hat{t} \le t$.
So $p'$ $\exs$ line~\ref{cleanr2-2} in round $j+2$ at some time $\hat{t'} <\hat{t} \le t$, i.e., $p'$ $\exs$ line~\ref{cleanr2-2} in round $j+2$ \emph{before} $p \in \{p_0 , p_1 \}$ enters round $2$ --- a contradiction to Lemma~\ref{bigarestuck}.
\end{proof}

Note that \emph{before} $p$ enters round~$j+2$ at time $t$,
	$p$ does the following in round $j+1$:
	it writes $0$ into $R_2$ in line~\ref{cleanr2} at some time $t_1$,
	it reads $R_2$ into $v$ in line~\ref{v1} at some time $t_2$,
	and then it finds that $v \ge n-2$ in line~\ref{guard2} (because $p$ does \emph{not} exit in line~\ref{exit2});
	clearly, $t_1 <t_2 <t$.
Thus, some process $p_j \in \{ \procs \}$ must have written a value $u \ge n-2$ into $R_2$ in line~\ref{pwrite2}
	at some time $t' \in [t_1 , t_2]$ in some round~$r$.
Since $t' \le t_2 <t$, by Claim~\ref{nc}, $r \le j+1$.

\begin{claim}\label{aboutW}
$r= j+1$.
\end{claim}

\begin{proof}
Suppose, for contradiction, $r \le j$.
Thus $p_j$ writes $R_2$ in line~\ref{pwrite2} in round $r \le j$ at time $t' \in [t_1 , t_2]$, \MP{We mean that $p_j$'s write $R_2$ in line~\ref{pwrite2} must take effect in that interval. The invocation of this write can occur before $t_1$!!!}
	i.e.,  \emph{after} the time $t_1$
	when $p$ writes $0$ into $R_2$ in line~\ref{cleanr2}
	in round $j+1$.
So $p_j$ writes $R_2$ in line~\ref{pwrite2} in round $r \le j$ \emph{after} $p \in \{ p_0 , p_1 \}$ enters round $r+1 \le j+1$
	--- a contradiction to
	our Induction Hypothesis~($\star$).
\end{proof}

\noindent
From Claim~\ref{aboutW}: $p_j$ writes $u \ge n-2$ into $R_2$ in line~\ref{pwrite2} in round $j+1$
	at time $t' \in [t_1 , t_2]$ ($\dagger$).

\begin{claim}\label{needaproof}
If a process in $\{ \procs \}$ writes $k\ge1$ into $R_2$ in line~\ref{pwrite2} in round $j+1$ at some time $\hat{t} \le t_2$,
then at least $k$ distinct processes in $\{ \procs \}$ write $R_2$ in line~\ref{pwrite2} in round $j+1$~by~time~$\hat{t}$.
\end{claim}

\begin{proof}
The proof is by induction on $k$.

\begin{itemize}
\item \textsc{Base Case:} For $k=1$ the claim trivially holds.

\item \textsc{Induction Step:} Let $k \ge1$.
Induction Hypothesis ($\mathsection \mathsection$):
	if a process in $\{ \procs \}$ writes $k$ into $R_2$ in line~\ref{pwrite2} in round $j+1$
	at some time $\hat{t} \le t_2$,
	then at least $k$ distinct processes in $\{ \procs \}$ write $R_2$ in line~\ref{pwrite2} in round $j+1$ by time $\hat{t}$.

Suppose that a process $p_{\ell} \in \{ \procs \}$ writes $k+1$ into $R_2$ in line~\ref{pwrite2} in round $j+1$ at some time $\twkp \le t_2$,
we must show that least $k+1$ distinct processes in $\{ \procs \}$ write $R_2$ in line~\ref{pwrite2} in round $j+1$ by time $\twkp$.

Since $p_{\ell}$ writes $k+1$ into $R_2$ in line~\ref{pwrite2} in round $j+1$ at some time $\twkp \le t_2$:

\begin{enumerate}

\item $p_{\ell} $ writes $0$ into $R_2$ in line~\ref{cleanr2-2} in round $j+1$ at some time $\two$, and

\item $p_{\ell}$ reads $k$ from $R_2$ in line~\ref{pread2} in round $j+1$ at some time $\trk$,
	such that $\two < \trk < \twkp \le t_2 < t$.

\end{enumerate}

\noindent
Thus, some process $p^* \in \{ \procs \}$ must have written $k$ into $R_2$
	in line~\ref{pwrite2} at some time $t^* \in [\two, \trk]$ in some round $r$.
Since $t^* < t$, by Claim~\ref{nc}, $r \le j+1$.

\begin{cclaim}\label{phew}
$r= j+1$.
\end{cclaim}
\vspace{-8pt}
\begin{proof}
Suppose, for contradiction, that $r \le j$.
Since $p_j$ $\exs$ line~\ref{cleanr2-2} in round $j+1$ at time~$\two$,
	by Lemma~\ref{bigarestuck}, both $p_0$ and $p_1$ entered round $j+1$ \emph{before} time $\two$.
So $p_0$ and $p_1$ enter round $r+1\le j+1$ \emph{before} $p^*$ writes $R_2$ in line~\ref{pwrite2} in round $r \le j$
	(at time $t^*$) --- a contradiction to our Induction Hypothesis~($\star$).
\end{proof}

From Claim~\ref{phew}, $p^*$ writes $k$ into $R_2$ in line~\ref{pwrite2} in round $j+1$ at time $t^* \in [\two, \trk]$.
Since $t^* \le \trk <t_2$, by the Induction Hypothesis ($\mathsection \mathsection$),
	at least $k$ distinct processes in $\{ \procs \}$ write $R_2$ in line~\ref{pwrite2} in round $j+1$ by time $t^*$.
Recall that $p_{\ell} \in \{ \procs \}$ also writes $R_2$ in line~\ref{pwrite2} in round $j+1$ at time $\twkp$.
Thus, since each process in $\{ \procs \}$ writes $R_2$ in line~\ref{pwrite2} in round $j+1$ \emph{at most once},
	at least $k+1$ distinct processes in $\{ \procs \}$ write $R_2$ in line~\ref{pwrite2}
	in round $j+1$ by time~$\twkp$.
\end{itemize}
\end{proof}

\noindent
Recall that by ($\dagger$), $p_j$ writes $u \ge n-2$ into $R_2$ in line~\ref{pwrite2} in round $j+1$
	at time $t' \in [t_1 , t_2]$.
Since $t' \le t_2$, by Claim~\ref{needaproof},
	at least $n-2$ distinct processes in $\{ \procs \}$ write $R_2$ in line~\ref{pwrite2} in round $j+1$ by time $t'$.
Thus \emph{every} process in $\{ \procs \}$ writes $R_2$ in line~\ref{pwrite2} in round $j+1$ by time~$t'$.
Since $t'\le t_2 < t$, every $p_i \in \{ \procs \}$ writes $R_2$ in line~\ref{pwrite2} in round $j+1$ before time $t$.
\end{proof}


\begin{restatable}{lemma}{safetyfour}\label{C-non-bot}
For all $j\ge1$,
	if $p_i \in \{ \procs \}$ $\exs$ line~\ref{guard1} in round $j$,
	then in that line $p_i$ has $c=b$ such that $b \in \{0,1\}$ and
	$p_0$ wrote $b$ into register $C$ in line~\ref{pcoin1-2} in round $j$.
\end{restatable}

\begin{proof}
Suppose that some $p_i \in \{ \procs \}$ $\exs$ line~\ref{guard1} in round $j \ge 1$.
We must show that  $p_i$ has $c=b$ for some $b \in \{0,1\}$ in line~\ref{guard1} in round $j \ge 1$,
	and $p_0$ wrote $b$ into $C$ in line~\ref{pcoin1-2} in round $j$.

Since $p_i$ $\exs$ line~\ref{guard1} in round $j$,
	in each round $k$, $1 \le k \le j$, the following occurs:
	
\begin{enumerate}
\item $p_i$ first writes $\bot$ into register $C$ in line~\ref{cleanc}, then
\item $p_i$ reads some value $b_k$ from $C$ into $c$ in line~\ref{rcoin1}.
\end{enumerate}

Note that $b_k \neq \bot$, because otherwise $p_i$
	would exit in line~\ref{exit0} (by the condition of line~\ref{guard0}) in round $k$ before reaching line~\ref{guard1} in round $j$.
Moreover, since $b_k \neq \bot$ and $p_0$ is the only process that writes non-$\bot$ values (namely, $0$ or $1$) into $C$,
	$b_k$ must be a value in $\{0,1\}$ that $p_0$ wrote into $C$ (in line~\ref{pcoin1-2}).
Therefore: in each round $k$, $1 \le k \le j$,
	$p_i$ writes $\bot$ into $C$ at some time $t_k$, and
	at some time $t'_k >t_k$, $p_i$ reads from $C$ a value $b_k \in \{0,1\}$ written by $p_0$ into $C$ at some time in $[t_k , t'_k]$~(*).

Since $p_0$ writes a value into $C$ \emph{only once} in each round,
	(*) implies that the value $b_j$ that $p_i$ reads from $C$ in round $j$ was written by $p_0$ into $C$ in some round $r \ge j$~(**).
	
Let $t$ be the time $p_i$ reads $b_j$ from $C$ in round $j$ (in line~\ref{rcoin1}).
At time $t$, $p_i$ has not yet reached line~\ref{pwrite2} in round $j$.
Thus, from Lemma~\ref{smallarestuck}, process $p_0$ has not entered round $j+1$ by time $t$.
So, by time $t$, process $p_0$ has not written any value into register $C$ in any round $\ell \ge j+1$.
Therefore, by (**), the value $b_j$ that $p_i$ reads from $C$ into $c$ in line~\ref{rcoin1} in round $j$ at time $t$
	was written by $p_0$ into $C$ in round $r = j$.
We conclude that $p_i$ has $c=b_j$ for $b_j \in \{0,1\}$ in line~\ref{guard1} in round $j$,
	and $p_0$ wrote $b_j$ into $C$ in line \ref{pcoin1-2} in~round~$j$.
\end{proof}

We now prove that if registers $R_1$, $R_2$, and $C$ are $\sly$,
	then Algorithm~\ref{toyalgo} terminates with \mbox{probability~1}
	even against a strong adversary (Theorem~\ref{WSLinearizableIsStrong}).
Intuitively,
	this is because if $R_1$, $R_2$, and $C$ are $\sly$,
	then the order in which $[0,j]$ and $[1,j]$ are written into $R_1$ in line~\ref{pwrite1} in round $j$ is already \emph{fixed}
	before the adversary $\mathcal{S}$ can see the result of the coin flip in line~\ref{pcoin1-1} in round $j$.
So for every round $j\ge 1$,
	the adversary can\emph{not} ``retroactively'' decide on this linearization order of write operations
	according to the coin flip result
	(like it did when $R_1$ was merely linearizable)
	to ensure that
	processes $p_1, p_2,...,p_{n-1}$ do not exit by the condition of line~\ref{guard1}.
Thus, with probability 1/2, all these processes will exit in line~\ref{exit1}.
And if they all exit there, then no process will increment register $R_2$ in lines~\ref{pread2}-\ref{pwrite2},
	and so $p_0$ and $p_1$ will also exit.

The proof of Theorem~\ref{WSLinearizableIsStrong} is based on the following:

\begin{lemma}\label{exitchance}
For all rounds $j \ge1$, 
	with probability at least $1/2$,
	no process enters round $j+1$.
\end{lemma}

\begin{proof}
Consider any round $j\ge1$.
There are two cases:

\begin{enumerate}[(I)]
\item\label{noRwrite} \textbf{Process $p_0$ does not complete its write of $[0,j]$ into register $R_1$ in line~\ref{pwrite1} in round $j$.}

Thus, $p_0$ does not invoke the write of any value into $C$ in line~\ref{pcoin1-2} in round $j$ (*).
 
\begin{claim}\label{allarestuck-det}
No process enters round $j+1$.
\end{claim}

\begin{proof}
We first show that no process in $\{ \procs \}$ $\exs$ line~\ref{pwrite2} in round $j$.
To see why, suppose, for contradiction,
	some process $p_i$ with $i\in\{2,3,...,n-1\}$
	$\exs$ line~\ref{pwrite2} in round $j$.
By Lemma~\ref{C-non-bot},
	in that line $p_i$ has $c=b \in \{0,1\}$ such that
	$p_0$ invoked the write of $b$ into $C$ in line~\ref{pcoin1-2} in round $j$ --- a contradiction to (*).

Thus no process in $\{ \procs \}$ $\exs$ line~\ref{pwrite2} in round $j$.
By Lemma~\ref{smallarestuck}, neither $p_0$ nor \mbox{$p_1$~enters~round~$j+1$}.
\end{proof}

\item\label{Rwrite} \textbf{Process $p_0$ completes its write of $[0,j]$ into register $R_1$ in line~\ref{pwrite1} in round $j$.}

\begin{claim}\label{allarestuck-rndm}
With probability at least 1/2, no process enters round $j+1$.
\end{claim}

\begin{proof}
Consider the set of histories $\seth$ of Algorithm~\ref{toyalgo}; this is a set of histories over
	the registers~$R_1$,~$R_2$,~$C$.
Since these registers are $\sly$,
	by Lemma 4.8 of~\cite{sl11}, 
	$\seth$ is $\sly$, i.e., it has
	at least one $\str$ function
	that satisfies properties (L) and (P) of Definition~\ref{defwsl}.
Let $f$ be the $\str$ function that
	the~adversary~$\mathcal{S}$~uses.

Let $g$ be an arbitrary history of the algorithm up to
	and including the completion of the write of $[0,j]$ into $R_1$
	by $p_0$ in line~\ref{pwrite1} in round $j$.
Since $p_0$ completes its write of $[0,j]$ into $R_1$ in $g$,  
	this write operation appears in the $\str$ $f(g)$.
Now there are two cases:

\begin{itemize}
\item\label{Casino1} \textbf{Case A}:
In $f(g)$, the write of $[1,j]$ into $R_1$ by $p_1$ in line~\ref{pwrite1} in round $j$
	occurs \emph{before}
	the write of $[0,j]$ into $R_1$ by $p_0$ in line~\ref{pwrite1} in round $j$.

Since $f$ is a $\str$ function,
	for every extension $h$ of the history $g$
	(i.e., for every history $h$ such that $g$ is a prefix of $h$),
	the write of $[1,j]$ into $R_1$
	occurs before
	the write of $[0,j]$ into $R_1$
	in the linearization $f(h)$
	(note that for all $j\ge1$, each of $[1,j]$ and $[0,j]$ is written \emph{at most once} in $R_1$, so it appears at most once in $h$ and $f(h)$).
Thus, in $g$ and every extension $h$ of~$g$,
	no process can first read $[0,j]$ from $R_1$ and then read $[1,j]$ from $R_1$ ($\star$).

Let $\mathcal{P}$ be the subset of processes in $\{p_2 , p_3, \ldots , p_{n-1} \}$
	that
	evaluate the condition
	($u_1 \neq [c,j]$ \textbf{or}  $u_2 \neq [1-c,j]$)
	in line~\ref{guard1} in round $j$.
Note that for each process $p_i$ in $\mathcal{P}$,
	$u_1$ and $u_2$ are the values that $p_i$ read from $R_1$
	consecutively in lines~\ref{u1} and \ref{u2} in round $j$.
By~($\star$), $p_i$ cannot first read $u_1=[0,j]$ and then read $u_2=[1,j]$ from $R_1$.
Thus, no process $p_i$ in $\mathcal{P}$ can have both $u_1=[0,j]$ and $u_2=[1,j]$ in line~\ref{guard1} in round $j$ ($\star \star$).

Let $\mathcal{P' \subseteq P}$ be the subset of processes in $\mathcal{P}$
	that have $c = 0$
	in line~\ref{guard1} in round $j$.

\begin{cclaim}\label{Mannaggia1}
~
\begin{enumerate}[(a)]
\item\label{SC1} No process in $\mathcal{P'}$ $\exs$ line~\ref{pwrite2} in round $j$.
\item\label{SC2} If $\mathcal{P'} = \mathcal{P}$ then neither $p_0$ nor $p_1$ enters round $j+1$.
\end{enumerate}	
\end{cclaim}

\begin{proof}
To see why Part~\ref{SC1} holds, note that no process $p_i$ in $\mathcal{P'}$
	can find the condition~($u_1 \neq [c,j]$ \textbf{or}  $u_2 \neq [1-c,j]$) in line~\ref{guard1} in round $j$
	to be false: otherwise $p_i$ would have both $u_1 = [c,j] = [0,j]$ \textbf{and} $u_2 = [1-c,j] = [1,j]$ in that line,
	but this is not possible by ($\star \star$).
Thus, no process in $\mathcal{P'}$ $\exs$ line~\ref{pwrite2} in round $j$
	(it would exit in line~\ref{exit1} before reaching that line).
	
To see why Part~\ref{SC2} holds, suppose $\mathcal{P'} = \mathcal{P}$
	and consider any process $p_i$ in $\{ \procs \}$.
If $p_i \not \in \mathcal{P}$
	then $p_i$ never evaluates the condition
	in line~\ref{guard1} in round $j$;
	and if $p_i \in \mathcal{P}$, then $p_i \in \mathcal{P'}$, and so from Part~\ref{SC1},
	$p_i$ does not $\ex$ line~\ref{pwrite2} in round $j$.
So \emph{in both cases},
	$p_i$ does not $\ex$ line~\ref{pwrite2} in round $j$.
Thus, by Lemma~\ref{smallarestuck}, neither $p_0$ nor $p_1$ enters round $j+1$.
\end{proof}

Now recall that $g$ is the history of the algorithm up to
	and including the completion of the write of $[0,j]$ into $R_1$
	by $p_0$ in line~\ref{pwrite1} in round $j$.
After the completion of this write, i.e., in any extension $h$ of $g$,
	$p_0$ is supposed to flip a coin and write the result into $C$ in line~\ref{pcoin1-2} in round~$j$.
Thus, with probability \emph{at least}~$1/2$,
	$p_0$~will \emph{not} invoke
	the operation to write $1$ into $C$ in line~\ref{pcoin1-2} in round~$j$.\MP{The adversary may decide to kill $p_0$
	before he flips a coin or before the write of $C$ in round $j$, or even after $p_0$ invokes the write  of $C$ in round $j$
	but before completing the write.}
So, from Lemma~\ref{C-non-bot},
	with probability \emph{at least}~$1/2$,
	every process in $\mathcal{P}$ has $c = 0$
	in line~\ref{guard1} in round $j$;
	this means that
        with probability at least~$1/2$,
	$\mathcal{P' = P}$.
Therefore, from Claim~\ref{Mannaggia1}, with probability at least~$1/2$:
\begin{enumerate}[(a)]
\item No process in $\mathcal{P}$ $\exs$ line~\ref{pwrite2} in round $j$.
\item Neither $p_0$ nor $p_1$ enters round $j+1$.
\end{enumerate}
This implies that in Case~A, with probability at least~$1/2$,
	no process enters round $j+1$.

\item\label{Casino2} \textbf{Case B}:
In $f(g)$, the write of $[1,j]$ into $R_1$ by $p_1$ in line~\ref{pwrite1} in round $j$
	does \emph{not} occur before
	the write of $[0,j]$ into $R_1$ by $p_0$ in line~\ref{pwrite1} in round $j$.
This case is essentially symmetric to the one for Case A, we include it below for completeness.

Since $f$ is a $\str$ function,
	for every extension $h$ of the history $g$,
	the write of $[1,j]$ into $R_1$
	does not occur before
	the write of $[0,j]$ into $R_1$
	in the linearization $f(h)$.
Thus, in $g$ and every extension $h$ of $g$,
	no process can first read $[1,j]$ from $R_1$ and then read $[0,j]$ from $R_1$ ($\dagger$).

Let $\mathcal{P}$ be the subset of processes in $\{p_2 , p_3, \ldots , p_{n-1} \}$
	that
	evaluate the condition
	($u_1 \neq [c,j]$ \textbf{or}  $u_2 \neq [1-c,j]$)
	in line~\ref{guard1} in round $j$.
Note that for each process $p_i$ in $\mathcal{P}$,
	$u_1$ and $u_2$ are the values that $p_i$ read from $R_1$
	consecutively in lines~\ref{u1} and \ref{u2} in round $j$.
By~($\dagger$), $p_i$ cannot first read $u_1=[1,j]$ and then read $u_2=[0,j]$ from $R_1$.
Thus, no process $p_i$ in $\mathcal{P}$ can have both $u_1=[1,j]$ and $u_2=[0,j]$ in line~\ref{guard1} in round $j$ ($\dagger \dagger$).

Let $\mathcal{P' \subseteq P}$ be the subset of processes in $\mathcal{P}$
	that have $c = 1$
	in line~\ref{guard1} in round $j$.

\begin{cclaim}\label{Mannaggia2}
~
\begin{enumerate}[(a)]
\item\label{SC21} No process in $\mathcal{P'}$ $\exs$ line~\ref{pwrite2} in round $j$.
\item\label{SC22} If $\mathcal{P'} = \mathcal{P}$ then neither $p_0$ nor $p_1$ enters round $j+1$.
\end{enumerate}	
\end{cclaim}

\begin{proof}
To see why \ref{SC21} holds, note that no process $p_i$ in $\mathcal{P'}$
	can find the condition ($u_1 \neq [c,j]$ \textbf{or}  $u_2 \neq [1-c,j]$) in line~\ref{guard1} in round $j$
	to be false: otherwise $p_i$ would have both $u_1 = [c,j] = [1,j]$ \textbf{and} $u_2 = [1-c,j] = [0,j]$ in that line,
	but this is not possible by ($\dagger \dagger$).
Thus, no process in $\mathcal{P'}$ $\exs$ line~\ref{pwrite2} in round $j$
	(it would exit in line~\ref{exit1} before reaching that line).
	
To see why \ref{SC2} holds, suppose $\mathcal{P'} = \mathcal{P}$
	and consider any process $p_i$ in $\{ \procs \}$.
If $p_i \not \in \mathcal{P}$
	then $p_i$ never evaluates the condition
	in line~\ref{guard1} in round $j$;
	and if $p_i \in \mathcal{P}$, then $p_i \in \mathcal{P'}$, and so from \ref{SC21},
	$p_i$ does not $\ex$ line~\ref{pwrite2} in round $j$.
So \emph{in both cases},
	$p_i$ does not $\ex$ line~\ref{pwrite2} in round $j$.
Thus, by Lemma~\ref{smallarestuck}, neither $p_0$ nor $p_1$ enters round $j+1$.
\end{proof}

Now recall that $g$ is the history of the algorithm up to
	and including the completion of the write of $[0,j]$ into $R_1$
	by $p_0$ in line~\ref{pwrite1} in round $j$.
After the completion of this write, i.e., in any extension $h$ of $g$,
	$p_0$ is supposed to flip a coin and write the result into $C$ in line~\ref{pcoin1-2} in round~$j$.
Thus, with probability \emph{at least}~$1/2$,
	$p_0$~will \emph{not} invoke
	the operation to write $0$ into $C$ in line~\ref{pcoin1-2} in round~$j$.
So, from Lemma~\ref{C-non-bot},
	with probability \emph{at least}~$1/2$,
	every process in $\mathcal{P}$ has $c = 1$
	in line~\ref{guard1} in round $j$;
	this means that
        with probability at least~$1/2$,
	$\mathcal{P' = P}$.
Therefore, from Claim~\ref{Mannaggia2}, with probability at least~$1/2$:
\begin{enumerate}[(a)]
\item No process in $\mathcal{P}$ $\exs$ line~\ref{pwrite2} in round $j$.
\item Neither $p_0$ nor $p_1$ enters round $j+1$.
\end{enumerate}
This implies that in Case~B, with probability at least~$1/2$,
	no process enters round $j+1$.
\end{itemize}

So in both Cases A and B, with probability at least 1/2, no process enters round $j+1$.
\end{proof}
\end{enumerate}

\noindent
Therefore, from Claims~\ref{allarestuck-det} and~\ref{allarestuck-rndm} of Cases~\ref{noRwrite} and~\ref{Rwrite},
	with probability at least $1/2$,
	no process enters round $j+1$.
\end{proof}

\noindent
We can now complete the proof of Theorem~\ref{WSLinearizableIsStrong},
	namely, that with $\sly$ registers,
	Algorithm~\ref{toyalgo} terminates with probability 1 
	even against a strong adversary.

\toyterminate*


\begin{proof}
Consider any round $j \ge 1$.
By Lemma~\ref{exitchance},
	with probability at least $1/2$,
	no process enters round $j+1$.
Since this holds for every round $j \ge 1$,
	then, with probability 1,
	all the correct processes
	return in lines~\ref{halt0} or line~\ref{halt1}
	within a finite number of rounds $r$.
\end{proof}

\section{Bounding the registers of Algorithm~\ref{toyalgo} }\label{btoysection}
We now explain how to obtain our result of Section~\ref{sectiontoyalgo},
	namely Theorems~\ref{LinearizableIsWeak}--\ref{WSLinearizableIsStrong},
	 with \emph{bounded} shared registers.
Note that Algorithm~\ref{toyalgo} that we used to obtain our results
	uses three shared MWMR registers, and only one of them, namely register $R_1$, is unbounded.
Specifically, $R_1$ is unbounded because
	each process $p_i$, $i\in \{0, 1\}$, writes a tuple $[i,j]$ into $R_1$ in rounds $j = 1, 2, \ldots$.
As we show below, however, it is sufficient for $R_1$ to contain only the values $0$, $1$, or $\bot$.
To achieve this, we modify Algorithm~\ref{toyalgo} as follows:

\begin{itemize}

\item Each process $p_i$, $i\in \{0, 1\}$,
	writes $i$ (instead of the tuple $[i,j]$) into $R_1$ in line~\ref{pwrite1}. 

\item The guard ($u_1 \neq [c,j]$ or $u_2 \neq [1-c,j]$)
	is replaced by the guard ($u_1 \neq c$ or $u_2 \neq 1-c$) in line~\ref{guard1}.  

\end{itemize}

To understand why the above changes do not affect the behaviour of the algorithm
	(i.e., why Algorithm~\ref{toyalgo} and the modified algorithm have exactly the same runs),
	consider how the tuples $[i,j]$ written into $R_1$ are actually used in Algorithm~\ref{toyalgo}.
In each round $j \ge 1$, each process $p \in \{\procs \}$ reads $R_1$ twice (in lines~\ref{u1} and in line~\ref{u2}, respectively),
	and then it checks the two values $u_1$ and $u_2$ that it read
	in line~\ref{guard0} and in line~\ref{guard1}.
Specifically:

\begin{enumerate}

	\item In line~\ref{guard0}, $p$ checks whether $u_1$ or $u_2$ is $\bot$.
	
	Note that $p$ has $u_1 = [i,j]$ $(\neq \bot)$ in line~\ref{guard0} in Algorithm~\ref{toyalgo}
	if and only if
	$p$ has $u_1 = i$ $(\neq \bot)$ in line~\ref{guard0} in the modified algorithm.
	The same holds for $u_2$ in line~\ref{guard0}.

	Thus, the guard ($u_1 = \bot$ or $u_2 = \bot$ or $c = \bot$) in line~\ref{guard0} has the same effect
	in the modified algorithm as in Algorithm~\ref{toyalgo}.
	
	\item In line~\ref{guard1}, $p$ checks whether ($u_1 \neq [c,j]$ or $u_2 \neq [1-c,j]$) for some specific value $c$.
	
	Lemma~\ref{mini} (proven below) states that when
	$p$ reaches line~\ref{guard1} in round $j$
	and is poised to check whether ($u_1 \neq [c,j]$ or $u_2 \neq [1-c,j]$),
	it must be the case that $u_1 = [-,j]$ \emph{and} $u_2 = [-,j]$.
	So checking whether the second component of the two tuples is $j$ or not is useless.
	Thus, the guard ($u_1 \neq [c,j]$ or $u_2 \neq [1-c,j]$)
	in Algorithm~\ref{toyalgo}
	has the same effect as the guard ($u_1 \neq c$ or $u_2 \neq 1-c$) in the modified algorithm.

\end{enumerate}

So to show that the modified algorithm
	behaves exactly as Algorithm 1,
	 it now is sufficient to prove:
\noindent
\begin{lemma}\label{mini}
For all $j\ge1$,
	if $p_i \in \{ \procs \}$ $\exs$ line~\ref{guard1} in round $j$,
	then in that line $p_i$ has $u_1 = [-,j]$ and $u_2 = [-,j]$. 
\end{lemma}

\begin{proof}
Suppose that some $p_i \in \{ \procs \}$ $\exs$ line~\ref{guard1} in round $j \ge 1$.
We must show that $p_i$ has $u_1 = [-,j]$ and $u_2 = [-,j]$ in line~\ref{guard1} in round $j$.

Since $p_i$ $\exs$ line~\ref{guard1} in round $j$,
	$u_1 \neq \bot$ and $u_2 \neq \bot$ in that line, because otherwise $p$
	would have exited in line~\ref{exit0} (by the condition of line~\ref{guard0}) before reaching line~\ref{guard1} in round $j$.
Note that
	$u_1$ and $u_2$ contain the values
	that $p_i$ reads from $R_1$ in lines~\ref{u1} and \ref{u2} in round~$j$,
	after it wrote $\bot$ into $R_1$ in line~\ref{cleanr1} in round~$j$.
So 
	in line~\ref{guard1} in round $j$ of $p_i$,
	each of $u_1$ and $u_2$ contains a non-$\bot$ value that some process in $\{ p_0 , p_1 \}$
	wrote into $R_1$ (in line~\ref{pwrite1})
	between the time $t_w$ when $p_i$ wrote $\bot$ into $R_1$ in line~\ref{cleanr1} in round~$j$
	and the time $t_r$ when $p_i$ read $R_1$ in
	line~\ref{u2} in round~$j$ (*).
Let $p_b$ be any process in $\{ p_0 , p_1 \}$ that writes $R_1$ (in line~\ref{pwrite1})
	at any time in $[t_w, t_r]$.
Note that $p_b$ does this write \emph{before} $p_i$ $\exs$ line~\ref{pwrite2} in a round $j$.
Thus, by Lemma~\ref{smallarestuck}, $p_b$ does this write \emph{in a round $r \le j$}.
So this is a write of $[b,r]$ with $r \le j$ and $b \in \{0, 1\}$ into~$R_1$.
Thus, from (*), $p_i$ has $u_1 = [b_1,r_1]$ and $u_2 = [b_2,r_2]$ with $1 \le r_1, r_2 \le j$ and $b_1, b_2 \in \{0, 1\}$
	in line~\ref{guard1} in round $j$ (**).

\begin{claim}\label{AlliMorte}
$r_1 = r_2 = j$
\end{claim}

\begin{proof}
Recall that $1 \le r_1, r_2 \le j$.
If $j =1$, it follows that $r_1 = r_2 = j =1$.
Now assume that $j \ge 2$.
Suppose, for contradiction, that $r_1 \neq j$ or $r_2 \neq j$,
	and so $1 \le r_1 \le  j-1$ or $1 \le r_2 \le  j-1$.
Without loss of generality,\MP{Perhaps this wlog is not obvious, $r_k, k\in\{1,2\}$ ?}
	say that $1 \le r_1 \le  j-1$.
By (**), $p_i$ has $u_1 = [b_1,r_1]$ with $b_1 \in \{0, 1\}$
	in line~\ref{guard1} in round~$j$.
So $p_i$ read $[b_1,r_1]$ from $R_1$ in lines~\ref{u1} in round~$j$,
	and $p_i$ did so \emph{after} writing $\bot$ into $R_1$ in line~\ref{cleanr1} in round~$j$.
Note that \emph{before} $p_i$ wrote $\bot$ into $R_1$ in line~\ref{cleanr1} in round~$j$, the following occurred:
	(i) $p_i$ reached line~\ref{pread2} \emph{in round $r_1 \le j-1$}; and so by Lemma~\ref{prev-reads},
	(ii) $p_i$ read the value $[b_1,r_1]$ from $R_1$ in round $r_1 \le j-1$.
From the code of processes $\{p_0, p_1 \}$ it is clear that $[b_1,r_1]$ is written \emph{only once}
	into~$R_1$ (specifically, by process $p_{b_1}$ in round $r_1$).
Thus, since $p_i$ read $[b_1,r_1]$ from $R_1$ in round $r_1 \le j-1$
	\emph{before} process $p_i$ writes
	$\bot$ into $R_1$ in line~\ref{cleanr1} in round~$j$,
	it can\emph{not} read $[b_1,r_1]$ again from $R_1$ in lines~\ref{u1} in round~$j$ --- a contradiction.
\end{proof}

From Claim~\ref{AlliMorte} and (**), process $p_i$ has $u_1 = [b_1,j]$ and $u_2 = [b_2,j]$ with $b_1, b_2 \in \{0, 1\}$
	in line~\ref{guard1} in round $j$.
\end{proof}

\section{Proof of Theorem~\ref{aiswsl}}\label{proofwsl}
We now prove Theorem~\ref{aiswsl} in Section~\ref{sectionwslalgo},
	namely, $\Awsl$ is a $\sly$ implementation of a MWMR register from SWMR registers.
In the following,
	we consider an arbitrary history $H \in \seth$ of $\Awsl$ and the history $\seq{}=f(H)$ constructed by $\f$ on input $H$.
We first show $f$ is a linearization function of $\seth$,
	i.e., $\seq{}$ satisfies properties~\ref{p1}--\ref{p3} of Definition~\ref{defl}.

\begin{definition}
An operation that starts at time $s$ and completes at time $f$ is active at time $t$ if $s \le t \le f$.
\end{definition}

\begin{definition}\label{tscompare}
Let $a$ and $b$ be two timestamps 
then:
\begin{itemize}

\item $a < b$ if and only if $a$ precedes $b$ in lexicographic order. 

\item $a \leq b$ if and only if $a = b$ or $a < b$.
\end{itemize}
\end{definition}

\begin{observation}
Relation $\le$ is a total order on the set of timestamps.  
\end{observation}

\begin{notation}
Let $w$ be an operation that writes into $Val[-]$.
We denote by $\ts{w}$ the timestamp that $w$ writes into $Val[-]$.
That is, if $w$ writes $(-,ts)$, 
	$\ts{w}=ts$.
\end{notation}

\begin{observation}\label{onlyone}
The tuples $(v,ts)$ and $(v',ts')$ written to $Val[-]$ by two distinct write operations have distinct timestamps, i.e.,
	$ts \ne ts'$.
\end{observation}

\begin{observation}\label{a1ob2}
Consider the execution of a write operation (lines~\ref{timestampb}--\ref{wreturn}) by a process $p_k$.
During that execution,
	the values of the variable $\nts$ of $p_k$ are non-increasing with time.
\end{observation}

\begin{observation}\label{obsreadtime}
		If a read operation returns $(v,ts) \neq (0,[0, \ldots, 0])$\XMP{If a read operation reads $(v,ts)$ such that $ts \neq [0, \ldots, 0]$ ?},
		then there is an operation $w$ 
		that writes $(v,ts)$ to $Val[-]$.
\end{observation}

\begin{lemma}\label{last}
If a read operations $r$ starts after an operation $w$ writes to $Val[-]$ and $r$ returns $(-,ts)$,
	then $ts \ge \ts{w}$.
\end{lemma}
\begin{proof} 
 Assume a read operations $r$ starts after an operation $w$ writes to some $Val[k]$ and $r$ returns $(-,ts)$.
Then $r$ reads $Val[i]$ (line~\ref{readall} of $\Awsl$) for every $1 \le i \le n$ after $w$ writes $(-,\ts{w})$ to $Val[k]$.
Since the timestamps in each $Val[-]$ are monotonically increasing, 
	$r$ reads $(-,ts'')$ from $Val[k]$ for some $ts' \ge \ts{w}$.
Since $ts$ is the largest timestamp that $r$ reads among all $Val[-]$ (lines~\ref{max}--\ref{readreturn} of $\Awsl$),
	$ts~\ge~ts'~\ge~\ts{w}$.
\end{proof}

\begin{observation}\label{w}
If an operation $w$ writes to $Val[-]$,
		there is an $i\ge 1$ such that 
		$w = w_i$.
\end{observation}

\begin{observation}\label{fmapwritex0}
If an operation $w$ writes to $Val[-]$,
		there is a unique $j\ge 1$ such that 
		$w \in \B{j}$.
\end{observation}

\begin{observation}\label{leaderofbatch0}
For every write operation $w$,
	$w \in \seq{}$ if and only if 
	there is an $i$ such that $w\in \mathcal{B}_i$.	
\end{observation}


By Observations~\ref{fmapwritex0} and~\ref{leaderofbatch0}, we have:

\begin{corollary}\label{wwi}
If an operation $w$ writes to $Val[-]$,
	then $w \in \seq{}$.
\end{corollary}


\begin{observation}\label{leaderofbatch1}
For any two write operations $w$ and $w'$,
	if $w\in \mathcal{B}_i$,	
	$w' \in \mathcal{B}_j$, 
	and $i < j$,
	then $w$ is before $w'$ in $\seq{}$.
\end{observation}

Recall that
	$\pts{w}{i}$ is the value of $\nts$,
	 at time $t_i$,
	 of the process executing the write operation $w$ (see line~\ref{newts} of $f$).

\begin{observation}\label{leaderofbatch2}
For all $i \ge 1$,
	if $w_i \in \B{i}$
	then $\pts{w_i}{i} = \ts{w_i}$.
\end{observation}

\begin{observation}\label{leaderofbatch3}
For all $i \ge 1$,
	 for all operations $w \in \B{i}$,
	 if $w \ne w_i $
	 then $\pts{w}{i} < \ts{w_i}$.
\end{observation}

By Observation~\ref{a1ob2}, we have:
\begin{observation}\label{leaderofbatch5}
For all $i \ge 1$,
	for all operations $w \in \B{i}$ that write to $Val[-]$,
	 $\ts{w} \le \pts{w}{i}$.
\end{observation}

By Observations~\ref{leaderofbatch3} and \ref{leaderofbatch5}, we have:
\begin{observation}\label{leaderofbatch4}
For all $i \ge 1$,
	for all operations $w \in \B{i}$ that write to $Val[-]$,
	 $\ts{w}  \le  \ts{w_i}$.
\end{observation}

\begin{lemma}\label{later}
For all $i \ge 1$,
	 for all operations $w$ that write to $Val[-]$,
	 if $w \in \sC{i}$ and $w \notin \B{i}$,
	 then $\ts{w} > \ts{w_i}$.\XMP{New lemma says that if $f$ decides not to linearize $w$ by time $t_i$, then $\ts{w} > \ts{w_i}$}
\end{lemma}
\begin{proof}
Let $i \ge 1$ and assume that an operation $w$ writes to $Val[-]$
	such that $w \in \sC{i}$ and $w \notin \B{i}$.	
By line~\ref{Bt} of $\f$,
	$\pts{w}{i} > \pts{w_i}{i}$.
Since by Observation~\ref{leaderofbatch2} $\pts{w_i}{i} = \ts{w_i}$,
	$\pts{w}{i} > \ts{w_i}$ (*).
There are two cases:

\textbf{Case 1:} $\pts{w}{i}$ contains no $\infty$.
This implies $\ts{w} = \pts{w}{i}$.
Thus, by (*),
	$\ts{w} = \pts{w}{i} > \ts{w_i}$.

\textbf{Case 2:} $\pts{w}{i}$ contains $\infty$.
Then there is a $k$ such that for every $k \le l \le n$,
	$w$ reads $(Val[l].ts)[l]$ (lines~\ref{ts1} and~\ref{ts2} of $\Awsl$) after time $t_i$.
Note that $w_i$ reads $(Val[l].ts)[l]$ before time $t_i$ for every $k \le l \le n$,
	and $w_i$ writes $(-,\ts{w_i})$ to some $Val[-]$ at time $t_i$.
Thus,
	for every $k \le l \le n$,
	since $(Val[l].ts)[l]$ is non-decreasing,
	$\ts{w}[l] \ge \ts{w_i}[l]$ ($\dagger$).

For every $1 \le l \le k-1$, 
	since $w$ reads $(Val[l].ts)[l]$ before time $t_i$,
	$\ts{w}[l] = \pts{w}{i}[l]$.
By (*), 
	 $\pts{w}{i}[1,\ldots, k-1] \ge \ts{w_i}[1,\ldots, k-1]$.
So $\ts{w}[1,\ldots, k-1] \ge \ts{w_i}[1,\ldots, k-1]$ ($\dagger \dagger$).

By ($\dagger$) and ($\dagger \dagger$),
	$\ts{w} \ge \ts{w_i}$.
Since $w \notin \B{i}$,
	$w \ne w_i$.
By Observation~\ref{onlyone},
	$\ts{w} \ne \ts{w_i}$.
Thus, $\ts{w} > \ts{w_i}$.
\end{proof}

\begin{lemma}\label{tsorder}\XMP{intuition of why we can decide the order of writes by seeing the partial new\_ts}
For all $j > i \ge 1$,
	if $w \in \B{i}$,
	$w' \in \B{j}$,
	and $w$ and $w'$ both write to $Val[-]$,
	then $\ts{w'} > \ts{w}$.
\end{lemma}

\begin{proof}
Assume $j > i \ge 1$,
	$w \in \B{i}$, 
	$w' \in \B{j}$,
	and $w$ and $w'$ both write to $Val[-]$.

\begin{claim}\label{tsorderc1}
$\ts{w'} > \ts{w_i}$.
\end{claim}

\begin{proof}
Since $w' \in \B{j}$ and $j > i$,
	$w' \notin \B{i}$\XMP{unique $\B{i}$ of Observation~\ref{fmapwritex0}} and $w' \notin \wseq{i}$.
There are two cases:

\textbf{Case 1}:
	$w' \in \sC{i}$.
Since $w' \in \sC{i}$ and $w' \notin \B{i}$,
	by Lemma~\ref{later},
	$\ts{w'} > \ts{w_i}$.
\smallskip

\textbf{Case 2}:
	$w' \notin \sC{i}$.
By lines~\ref{seqi1} and~\ref{seqi2} of $\f$,
	$\wseq{i-1}$ is a prefix of $\wseq{i}$.
Since $w' \notin \wseq{i}$, 
	$w' \notin \wseq{i-1}$.
Since $w' \notin \sC{i}$,
	by line~\ref{C} of $\f$,
	$w'$ is not active at time $t_i$.
Since $w' \in \B{j} \subseteq \sC{j}$,
 	by line~\ref{C} of $\f$,
 	$w'$ is active at time $t_j$.
Since $i < j$,
	$t_i < t_j$. 
So $w'$ starts after $t_i$
	and $w'$ reads all $(Val[-].ts)[-]$ (lines~\ref{timestampb}--\ref{timestampe} of $\f$) after $t_i$.
Note that $w_i$ reads $(Val[l].ts)[l]$ before time $t_i$ for every $1 \le l \le n$,
	and $w_i$ writes $(-,\ts{w_i})$ to some $Val[-]$ at time $t_i$.
Thus,
	for every $1 \le l \le n$,
	since $(Val[l].ts)[l]$ is non-decreasing,
	$\ts{w'}[l] \ge \ts{w_i}[l]$.
So $\ts{w'} \ge \ts{w_i}$.
Since $w' \notin \B{i}$,
	$w' \ne w_i$.
By Observation~\ref{onlyone},
	$\ts{w'} \ne \ts{w_i}$
	and so $\ts{w'} > \ts{w_i}$.

Therefore in both cases,
	$ts_{w'} > \ts{w_i}$.
\end{proof}

Since $w \in \B{i}$,
	by Observation~\ref{leaderofbatch4},
	$\ts{w} \le \ts{w_i}$.
By Claim~\ref{tsorderc1},
	$\ts{w'} > \ts{w_i} \ge \ts{w'}$.
\end{proof}

We now show that for every two operations $o_1$ and $o_2$, 
	if $o_1$ completes before $o_2$ starts in $H$\XMP{in $H$} and $o_1,o_2 \in \seq{}$,
	then $o_1$ is before $o_2$ in $\seq{}$.

\begin{lemma}\label{linearization1}
If a write operation $w$ completes before a write operation $w'$ starts and $w,w' \in \seq{}$,
	then $w$ is before $w'$ in $\seq{}$.
\end{lemma}

\begin{proof}
Assume a write operation $w$ completes before a write operation $w'$ starts and $w,w' \in \seq{}$.
By Observation~\ref{leaderofbatch0},
	there are $i$ and $j$ such that $w \in \B{i}$ and $w' \in \B{j}$.
By line~\ref{C} of $\f$,
	$w$ and $w'$ are active at time $t_i$ and $t_j$ respectively.
Since $w'$ starts after $w$ completes,
	$t_i < t_j$ and so $i < j$.
Thus, by Observation~\ref{leaderofbatch1},
	$w$ is before $w'$ in $\seq{}$.
\end{proof}

\begin{lemma}\label{last3}
If an operation $w$ writes to $Val[-]$ at time $t$ and $w \in \B{i}$ for some $i \ge 1$,
	then $t_i \le t$.
\end{lemma}
\begin{proof}
Assume,
	for contradiction,
	an operation $w$ writes to $Val[-]$ at time $t$,
	$w \in \B{i}$ for some $i \ge 1$,
	and $t_i > t$.
By Observation~\ref{w},
	there is a $k\ge 1$ such that $w = w_k$
	and $t = t_k$.
By lines~\ref{line13}--\ref{seqi1} of $\f$ 
	(the $k$th iteration of the for loop),
	$w_k \in \wseq{k}$.
Since $t_i > t = t_k$,
	$k \le i-1$.
By lines~\ref{seqi1} and~\ref{seqi2} of $\f$,
	$\wseq{k}$ is a prefix of $\wseq{i-1}$.
Since $w_k \in \wseq{k}$,
	 $w_k \in \wseq{i-1}$.
Thus,
	by line~\ref{C} of $\f$,
	$w_k \notin \sC{i}$
	and so $w_k \notin \B{i}$.
Since $w_k = w$,
	this contradicts that $w \in \B{i}$.
\end{proof}

\begin{lemma}\label{linearization2}
If a read operation $r$ completes before a write operation $w$ starts and $r, w \in \seq{}$,
	then $r$ is before $w$ in $\seq{}$.
\end{lemma}

\begin{proof}
Assume a read operation $r$ completes before a write operation $w$ starts and $r, w \in \seq{}$.
Let $(-,ts)$ denote the value that $r$ returns.

\textbf{Case 1}: $ts = [0, \ldots, 0]$.
By lines~\ref{readinit1}--\ref{readinit2} of $\f$,
	$r$ is before all the write operations in~$\seq{}$.
Thus,
	$r$ is before~$w$~in~$\seq{}$.

\textbf{Case 2}: $ts \ne [0, \ldots, 0]$.
By Observation~\ref{obsreadtime},
	an operation $w'$ writes $(-,ts)$ to $Val[-]$ and by Corollary~\ref{wwi}, 
	$w' \in \seq{}$.
By lines~\ref{followwrite1}--\ref{followwrite2} of $\f$,
	$r$ is after $w'$ and before any subsequent write in $\seq{}$.
So to show $r$ is before $w$ in $\seq{}$,
	it is sufficient to show that $w'$ is before $w$ in $\seq{}$.
\begin{claim}
 $w'$ is before $w$ in $\seq{}$.
\end{claim}
\begin{proof}
Since $w', w \in \seq{}$,
	by Observation~\ref{leaderofbatch0},
	there are $i$ and $j$ such that
	$w' \in \B{i}$ and	$w \in \B{j}$.
Let $t'$ be the time when $w'$ writes $(-,ts)$ to $Val[-]$.
Since $w' \in \B{i}$,
	by Lemma~\ref{last3},
	$t_i \le t'$.
Since $r$ returns $(-,ts)$,
	 $r$ reads $(-,ts)$ from $Val[-]$ and so $r$ completes after $t'\ge t_i$.
Since $w$ starts after $r$ completes,
	$w$ starts after time $t_i$.
Since $w \in \B{j} \subseteq \sC{j}$,
	by line~\ref{C} of $\f$,
	$w$ is active at time $t_j$.
Thus, $t_i < t_j$
	and so $i < j$. 
By Observation~\ref{leaderofbatch1},
	$w'$ is before $w$ in $\seq{}$.
\end{proof}
Therefore,
	in both cases 1 and 2,
	$r$ is before $w$ in $\seq{}$.
\end{proof}

\begin{lemma}\label{last1}
If a write operation $w$ writes to $Val[-]$ before a read operation $r$ starts and $w, r \in \seq{}$,
	then $w$ is before $r$ in $\seq{}$.
\end{lemma}

\begin{proof}
Assume a write operation $w$ writes to $Val[-]$ before a read operation $r$ starts and $w, r \in \seq{}$.
Let $(-,ts)$ 
	denote the value that $r$ returns.
By Lemma~\ref{last},
	$ts \ge \ts{w}$
	and so $ts \ne [0 \ldots 0]$.
By Observation~\ref{obsreadtime},
	an operation $w'$ writes $(-,ts)$\XMP{this notation is a bit weird since w' writes ts instead of ts'} to $Val[-]$
	and by Corollary~\ref{wwi},
	 $w' \in \seq{}$.
By lines~\ref{followwrite1}--\ref{followwrite2} of $\f$,
	$r$ is after $w'$ in $\seq{}$.
Since $ts \ge \ts{w}$,
	there are two cases:

\textbf{Case 1:} $ts = \ts{w}$.
By Observation~\ref{onlyone},
	$w = w'$.
So $r$ is after $w = w'$ in $\seq{}$.

\textbf{Case 2:} $ts > \ts{w}$.
Since $r$ is after $w'$ in $\seq{}$, 
	to show $w$ is before $r$ in $\seq{}$,
	it is sufficient to show that $w$ is before $w'$ in $\seq{}$.

\begin{claim}
$w$ is before $w'$ in $\seq{}$.
\end{claim}

\begin{proof}
Since $w,w' \in \seq{}$,
	by Observation~\ref{leaderofbatch0},
 	there are $i$ and $j$
 	such that $w \in \B{i}$ and $w' \in \B{j}$.

\textbf{Case (a)}: $i \ne j$.
Since $w \in \B{i}$,
	 $w' \in \B{j}$,
	 $i \ne j$,
	 and $ts > \ts{w}$,
	 by Lemma~\ref{tsorder},
	 $j > i$ (otherwise $ts < \ts{w}$). 
Thus,
	by Observation~\ref{leaderofbatch1},
	$w$ is before $w'$ in $\seq{}$.

\smallskip

\textbf{Case (b)}: $i = j$.
Then $w, w'\in\B{i}$.
\begin{cclaim}\label{wrcc1}
$w' = w_i$.
\end{cclaim}
\begin{proof}
Since $ w \in \B{i}\subseteq \sC{i}$,
	by line~\ref{C} of $\f$,
	$w$ is active at time $t_i$.
Since $r$ starts after $w$ completes,
	$r$ starts after $t_i$.
Then by Lemma~\ref{last}
	$ts \ge \ts{w_i}$.
Since $w' \in \B{i}$,
	by Observation~\ref{leaderofbatch4},
	$ts \le \ts{w_i}$.
Thus, $ts = \ts{w_i}$ and by Observation~\ref{onlyone},
	$w' = w_i$.
\end{proof}
Since $ts \ne \ts{w}$,
	$w' \ne w$.
Thus, $w \ne w_i$.
Since $w \in \B{i}$, 
	by line~\ref{seqi1} of $\f$,
	$w$ is before $w_i=w'$ in $\seq{}$.

Thus,
	 in both cases (a) and (b),
	  $w$ is before $w'$ in $\seq{}$.
\end{proof}	
Therefore, in both cases 1 and 2, 
	$w$ is before $r$ in $\seq{}$.
\end{proof}

Note that a completed write operation $w$ writes to $Val[-]$ before it completes.
Thus, 
	Lemma~\ref{last1} immediately implies the following: 
\begin{corollary}\label{linearization3}
If a write operation $w$ completes before a read operation $r$ starts and $w, r \in \seq{}$,
	then $w$ is before $r$ in $\seq{}$.
\end{corollary}

\begin{lemma}\label{linearization4}
If a read operation $r$ completes before a read operation $r'$ starts and $r,r' \in \seq{}$,
	then $r$ is before $r'$ in $\seq{}$.
\end{lemma}

\begin{proof}
Assume a read operation $r$ completes before a read operation $r'$ starts and $r,r' \in \seq{}$.
Let $(-,ts)$ denote the value that $r$ returns
and $(-,ts')$ denote the value that $r'$ returns.

\textbf{Case A:} $ts' = ts$.
By lines~\ref{R}--\ref{S} of $\f$,
	$r$ and $r'$ are in the same sequence $\SR$ that is ordered by their start time.
Since $r'$ starts after $r$ completes,
	$r$ is before $r'$ in $\SR$. 
By lines~\ref{readinit2} and \ref{followwrite2},
	$r$ is before $r'$ in $\seq{}$.

\textbf{Case B:} $ts' \ne ts$.

\textbf{Subcase B.1}: $ts = [0, \ldots, 0]$.
By lines~\ref{readinit1}--\ref{readinit2} of $\f$,
	$r$ is before all the write operations in $\seq{}$.
Since $ts' \ne ts$,
	$ts' \ne [0, \ldots, 0]$.
By Observation~\ref{obsreadtime},
	 an operation $w'$ writes $(-,ts')$ to $Val[-]$
	 and by Corollary~\ref{wwi}, 
	$w' \in \seq{}$.
By lines~\ref{followwrite1}--\ref{followwrite2} of $\f$,
	$r'$ is after $w'$ in $\seq{}$.
Since $r$ is before $w'$ in $\seq{}$,
	$r$ is before $r'$ in $\seq{}$.

\textbf{Subcase B.2}:
	$ts > [0, \ldots, 0]$.
By Observation~\ref{obsreadtime},
	 an operation $w$ writes $(-,ts)$ to $Val[-]$
	and by Corollary~\ref{wwi}, 
	$w \in \seq{}$.
Since $r$ reads $(-,ts)$ from $Val[-]$
	and $r'$ starts after $r$ completes,
	$r'$ starts after $w$ writes $(-,ts)$ to $Val[-]$.
By Lemma~\ref{last},
	$ts' \ge ts$.
So $ts' \ge ts > [0, \ldots, 0]$.
By Observation~\ref{obsreadtime},
		an operation $w'$ writes $(-,ts')$ to $Val[-]$ and by Corollary~\ref{wwi}, 
	$w' \in \seq{}$.
Since $ts' \ne ts$, $w' \ne w$. 
By lines~\ref{followwrite1}--\ref{followwrite2} of $\f$,
	 $r'$ is after $w'$ before any subsequent write in $\seq{}$
	 and $r$ is after $w$ before any subsequent write in $\seq{}$. 
Since $r'$ starts after $w$ writes to $Val[-]$,
	by Lemma~\ref{last1},
	$w$ is before $r'$ in $\seq{}$.
Thus,
	$w$ is before $w'$ in $\seq{}$
	and so $r$ is before $w'$ in $\seq{}$.
Since $r'$ is after $w'$ in $\seq{}$,
	$r$ is before $r'$ in $\seq{}$.

Therefore,
	in both cases A and B,
	$r$ is before $r'$ in $\seq{}$.	
\end{proof}

From Lemma~\ref{linearization1}, 
	Lemma~\ref{linearization2},
	Corollary~\ref{linearization3},
	and Lemma~\ref{linearization4},
	we have the following:
\begin{corollary}\label{linearizationf}
For every two operations $o_1$ and $o_2$, 
	if $o_1$ completes before $o_2$ starts and $o_1,o_2 \in \seq{}$,
	then $o_1$ is before $o_2$ in $\seq{}$.
\end{corollary}

\begin{lemma}\label{l1}
$\seq{}$ contains all completed operations of $H$ and possibly some non-completed ones. 
\end{lemma}
\begin{proof}
By Corollary~\ref{wwi},
	$\seq{}$ contains all the completed write operations in $H$ and possibly some non-completed ones.
Let $r$ be any completed read operation in $H$.
Since $r$ is completed,
	 $r$ returns some value $(-,ts)$.
If $ts = [0,...0]$,
	by line~\ref{readinit2} of $\f$,
	$r$ is in $\seq{}$.
If $ts \ne [0,...0]$,
	by Observation~\ref{obsreadtime} and Corollary~\ref{wwi},
	there is a write operation $w$ in $\seq{}$ that writes $(v,ts)$.
So by lines~\ref{followwrite1} and \ref{followwrite2} of $\f$,
	$r$ is in $\seq{}$.
In both cases, $\seq{}$ contains the completed read operation $r$.
Thus, $\seq{}$ contains all completed operations of $H$ and possibly some non-completed ones.
\end{proof}

\begin{observation}\label{reg}
For any read operation $r$ in $\seq{}$,
		if no write precedes $r$ in $\seq{}$, then $r$ returns the initial value of the register;
		Otherwise, $r$ returns the value written by the last write that occurs before $r$ in $\seq{}$
\end{observation}

\begin{lemma}\label{fsatl2}\XMP{linearization -> linearization function}
$f$ is a linearization function of $\seth$.
\end{lemma}


\begin{proof}
Recall that:
	(1)~$\seq{}$ is the output of $\f$ ``executed'' on an arbitrary history
	$H \in \seth$ of $\Awsl$ (the implementation on MRMW registers), and
	(2)~$\f$ defines the linearization function $f$; in other words, $f(H) = \seq{}$.
Furthermore,
	$\seq{}$ satisfies properties~\ref{p1}, \ref{p2}, and \ref{p3} of Definition~\ref{defl},
	by Lemma~\ref{l1},
	 Corollary~\ref{linearizationf},
	and Observation~\ref{reg},
	respectively. 
Thus, $f$ is a linearization function of $\seth$.
\end{proof}



To prove that $f$ is a $\str$ function for the set of histories $\seth$ of $\Awsl$,
	it now suffices to show that $f$ also satisfies property~(P) of Definition~\ref{defwsl}, namely:

\begin{restatable}{lemma}{propp}\label{fsatp}
For any $G, H \in \seth$, 
	if $G$ is a prefix of $H$, 
	then the sequence of write operations in $f(G)$\XMP{f(G) here is SEQ} 
	is a prefix of the sequence of write operations in $f(H)$.
\end{restatable}

\begin{proof}
Consider two histories $G \in \seth$ and $H \in \seth$ such that $G$ is a prefix of $H$.
If $G=H$, the lemma trivially holds.
Henceforth we assume that $G$ is a proper prefix of $H$, and so $G$ is finite.

In the following we use superscripts $G$ and $H$ to distinguish the value of the variables of $\f$ (that defines the linearization function $f$) when it is applied to input $G$ or $H$.
For example,
	$\wseq{i}^G$ denotes the value of $\wseq{i}$ of $\f$ on input $G$,
		and $\wseq{i}^H$ denotes the value of $\wseq{i}$ of $\f$ on input $H$.
Let $w_1 , w_2, \ldots, w_k$ be the operations that write to $Val[-]$ in history $G\in \seth$.
Since $G$ is a prefix of $H$, 
	then $w_1 , w_2, \ldots, w_k$ are also the first $k$ operations that write to $Val[-]$ in history $H \in \seth$.
Furthermore,
	for $1\le i \le k$,
	$w_i$ writes to $Val[-]$ at the same time in $G$ and $H$,
	say at time $t_i$.
\begin{claim}\label{seqgh}
	$\wseq{k}^G = \wseq{k}^H$.
\end{claim}

\begin{proof}
We use induction to prove that for all $0 \le i \le k$,
	$\wseq{i}^G = \wseq{i}^H$.
\begin{itemize}
\item \textsc{Base Case:} $i = 0$.
By line~\ref{seq0} of $\f$,
	$\wseq{0}^G = \wseq{0}^H = ()$.
The claim holds.

\item \textsc{Inductive Step:}
assume for some $0 \le i < k$,
	$\wseq{i}^G = \wseq{i}^H$ (IH).

\textbf{Case 1}:
	$w_{i+1} \in \wseq{i}^G$.
By (IH),
	$w_{i+1} \in \wseq{i}^H$.
By line~\ref{seqi2} of $\f$,
	$\wseq{i+1}^G= \wseq{i}^G$ and $\wseq{i+1}^H= \wseq{i}^H$.
So by (IH),
	$\wseq{i+1}^G=\wseq{i}^G= \wseq{i}^H= \wseq{i+1}^H$.

\textbf{Case 2}:
	$w_{i+1}\notin \wseq{i}^G$.
By (IH),
	$w_{i+1}\notin \wseq{i}^H$.
Since $G$ is a prefix of $H$,
	the set of write operations that are active at time $t_{i+1}$ is the same in $G$ and $H$. 
So by (IH) and line~\ref{C} of $\f$,
	$\sC{i+1}^G= \sC{i+1}^H$.
Furthermore,
	for all $w \in \sC{i+1}^G= \sC{i+1}^H$,
	$\pts{w}{i(G)} = \pts{w}{i(H)}$.
By line~\ref{Bt} of $\f$,
	$\B{i+1}^G= \B{i+1}^H$.
So the sequence of operations $w \in \B{i+1}^G$ in increasing order of $\pts{w}{i(G)}$ is equal to the sequence of operations $w \in \B{i+1}^H$ in increasing order of $\pts{w}{i(H)}$.
Thus, by (IH) and line~\ref{seqi1} of $\f$,
	$\wseq{i+1}^G= \wseq{i+1}^H$.

In both cases, the claim holds.
\end{itemize}
\end{proof}
Thus,
	since $\wseq{k}^H$ is a prefix of $\wseq{}^H$\XMP{make a new observation?},
	$\wseq{k}^G$ is a prefix of $\wseq{}^H$.
Since $w_k$ is the last operation that writes to $Val[-]$ in $G$,
	$\wseq{}^G= \wseq{k}^G$.
So $\wseq{}^G$ is a prefix of $\wseq{}^H$.
Note that $\wseq{}^G$ and $\wseq{}^H$ are the
	sequences of write operations in $f(G)$ and $f(H)$, respectively.
So the lemma holds.
\end{proof}

\rmv{
Intuitively, the above lemma holds, because $\f$ scans the input history $H$ (of Algorithm 1) \emph{by increasing time},
	and it linearizes the write operations as follows:
	when some write to $Val[-]$ occurs, say at time $t_i$,
	\emph{without looking ``ahead'' at what occurs in $H$ after time $t_i$},
	$\f$ linearizes a (possibly empty) batch of write operations
	that started before time $t_i$, 
	by \emph{appending} them to the sequence of write operations that it previously linearized.
The proof of Lemma~\ref{fsatp} is straightforward and relegated to Appendix~\ref{app1}.
}

By Lemmas~\ref{fsatl2}--\ref{fsatp},
	the function $f$ defined by $\f$ is a $\str$ function for the set of histories $\seth$
	of $\Awsl$.
Therefore:

\algowsl*



\section{Proof of Theorem~\ref{lamportisl}}\label{lamport1}
We now prove Theorem~\ref{lamportisl} in Section~\ref{sectionlamport},
	namely, $\Al$ is a linearizable implementation of a MWMR register from SWMR registers.
Let $\seth$ be the set of histories of the $\Al$. 
Consider an arbitrary history $H \in \seth$.
We first note that the timestamps of write operations respect the causal order of write events (they are Lamport clocks for these events).

\begin{lemma}\label{lampo}

If operations $w$ and $w'$ write $(-,ts)$ and $(-,ts')$ to $Val[-]$, respectively, and $w$ writes to $Val[-]$ before $w'$ starts, then $ts < ts'$.
\end{lemma}

\begin{proof}
Suppose operations $\W{}$ and $\W{}'$ write $(-,ts)$ and $(-,ts')$ to $Val[-]$ respectively and
	$\W{}$ writes to $Val[-]$ before $\W{}'$ starts.
Then $\W{}'$ reads all $Val[-]$ after $\W{}$ writes $(-,ts)$ to $Val[-]$.
Since the timestamp in each $Val[-]$ are non-decreasing,
	$\W{}'$ reads $(-,ts'')$ from $Val[-]$ for some $ts'' \ge ts$.
By lines~\ref{nts}--\ref{writeAl} of $\Al$,
	$ts'.sq > ts''.sq$ and so $ts' > ts''$.
So $ts \le ts'' < ts'$.
\end{proof}

\begin{observation}\label{onlyone2}
The tuples $(v,ts)$ and $(v',ts')$ written to $Val[-]$ by two distinct write operations have distinct timestamps, i.e.,
	$ts \ne ts'$.
\end{observation}

\begin{observation}\label{rwrelate}
If a read operation $\RD{}$ returns $(v,ts)$ such that $ts.sq \ne 0$,
	then there is a unique write operation~$\W{}$ that writes $(v,ts)$ to $Val[-]$.
\end{observation}

\begin{observation}\label{readinit}
If a read operation $\RD{}$ returns $(-,ts)$ such that $ts.sq = 0$
	then $ts = \langle 0,n\rangle$.
\end{observation}

Observation~\ref{readinit} implies:
\begin{observation}\label{readinit0}
If a read operation $\RD{}$ returns $(-,ts)$ and a read operation $\RD{}'$ returns $(-,ts')$ such that $ts' > ts$,
	then $ts'.sq > 0$.
\end{observation}

\begin{definition}\label{Lin}
Let $f$ be a function that maps the arbitrary history $H \in \seth$ of $\Al$
	to a sequential history $f(H) = \lS$ such that the following holds:

\begin{enumerate}[label=(\roman*)]
\item \label{allcw}  $\lS$ contains:

(a) all the write operations that write to $Val[-]$ (line~\ref{writeAl} of $\Al$) in $H$, and

(b) all the read operations that complete, i.e., return some value $(v,ts)$ (line~\ref{returnmax} of $\Al$) in~$H$.

\item \label{worder} If two write operations $\W{}$ and $\W{}'$ write
	$(-,ts)$ and $(-,ts')$ to $Val[-]$ in $H$ such that $ts <ts'$,
	then $\W{}$ is before $\W{}'$ in $\lS$.

\item\label{rorder1}  If a read operation $r$ returns some $(-,ts)$ with $ts.sq = 0$ in $H$,
then $r$ occurs in $\lS$ before every write operation $w$ in $\lS$.

\item\label{rorder2} If a read operation $r$ returns some $(-,ts)$ with $ts.sq \neq 0$ in $H$,
then $r$ occurs in $\lS$ as follows:

 (a) after the unique write operation $w$ that writes $(-,ts)$ in $\lS$
     (the operation $w$ is well-defined by Observation~\ref{rwrelate}), and
     
 (b) before every other subsequent write operation in $\lS$.

\item\label{rorder3}  If two read operations $r$ and $r'$ read some $(-,ts)$, and $r$ completes before $r'$ starts in $H$,
then $r$ occurs before $r'$ in $\lS$.

\end{enumerate}

\end{definition}


We now show that $f(H) = \lS$ satisfies properties~\ref{p1}, \ref{p2}, and \ref{p3} of Definition~\ref{defl},
	and so $f$ is a linearization function for the set of histories $\seth$ of $\Al$.

\begin{lemma}\label{l1_2}
$\lS$ contains all completed operations of $H$ and possibly some non-completed ones.
\end{lemma}

\begin{proof}
This follows immediately by~\ref{allcw} of Definition~\ref{Lin} and the fact that every completed write operation writes to $Val[-]$ (line~\ref{writeAl} of $\Al$) before it completes in $H$.
\end{proof}

By Lemma~\ref{lampo} and~\ref{worder} of Definition~\ref{Lin}, we have:
	
\begin{corollary}\label{wAw}
If a write operation $\W{}$ writes to $Val[-]$ before a write operation $\W{}'$ starts and $\W{},\W{}'\in \lS$
	then $\W{}$ is before $\W{}'$ in $\lS$.
\end{corollary}

\begin{corollary}\label{ww}
If a write operation $\W{}$ completes before a write operation $\W{}'$ starts and $\W{},\W{}'\in \lS$,
	then $\W{}$ is before $\W{}'$ in $\lS$.
\end{corollary}

\begin{lemma}\label{rw}
If a read operation $\RD{}$ completes before a write operation $\W{}$ starts and $\RD{},\W{} \in \lS$,
	then $\RD{}$ is before $\W{}$ in $\lS$.
\end{lemma}

\begin{proof}
Assume a read operation $\RD{}$ completes before a write operation $\W{}$ starts and $\RD{},\W{} \in \lS$.
Since $\W{} \in \lS$,
	by~\ref{allcw},
	$\W{}$ writes some $(-,ts)$ to $Val[-]$.
Let $(-,ts')$ denote the value that $\RD{}$ returns.

\textbf{Case 1:} $ts'.sq = 0$.
By~\ref{rorder1} of Definition~\ref{Lin},
	$\RD{}$ is before all the write operations in $\lS$.
So  $\RD{}$ is before $\W{}$ in $\lS$.

\textbf{Case 2:} $ts'.sq \ne 0$.
By~\ref{rorder2} of Definition~\ref{Lin},
	there is a unique write operation $\W{}' \in \lS$ that writes $(-,ts')$ to $Val[-]$
	such that $\RD{}$ is after $\W{}'$ and before any subsequent write operations in $\lS$.
Since $\RD{}$ reads $(-,ts')$ from $Val[-]$ and
	$\W{}$ starts after $\RD{}$ completes,
	$\W{}$ starts after $\W{}'$ writes $(-,ts')$ to $Val[-]$.
By Corollary~\ref{wAw},
	$\W{}'$ is before $\W{}$ in $\lS$.
So $\RD{}$ is before $\W{}$ in~$\lS$.
\end{proof}

\begin{lemma}\label{wr}
If a write operation $\W{}$ completes before a read operation $\RD{}$ starts and $\W{},\RD{} \in \lS$,
	then $\W{}$ is before $\RD{}$ in $\lS$.
\end{lemma}

\begin{proof}
Assume a write operation $\W{}$ completes before a read operation $\RD{}$ starts and $\W{},\RD{} \in \lS$.
Since $\W{}$ completes,
	$\W{}$ writes some $(-,ts)$ to $Val[-]$.
 Let  $(-,ts')$ denote the value that $\RD{}$ returns.
Since $\W{}$ completes before $\RD{}$ starts, 
	$\RD{}$ reads all $Val[-]$ after $\W{}$ writes $(-,ts)$ to $Val[-]$.
Since the timestamp in each $Val[-]$ are non-decreasing,
	$\RD{}$ reads $(-,ts'')$ from $Val[-]$ for some $ts'' \ge ts$.
By line~\ref{tsmax} of $\Al$,
	$ts'' \le ts'$.
So $ts \le ts'' \le ts'$.

\textbf{Case 1:} $ts = ts'$.
By Observation~\ref{onlyone2},
	$\RD{}$ returns the value that $\W{}$ writes.
So by~\ref{rorder2} of Definition~\ref{Lin},
	$\RD{}$ is after $\W{}$ in $\lS$.

\textbf{Case 2:} $ts < ts'$.
By Observation~\ref{readinit0},
	$ts'.sq > 0$.
By~\ref{rorder2} of Definition~\ref{Lin},
	some write operation $\W{}' \in \lS$ writes $(-,ts')$ to $Val[-]$ 
	and $\RD{}$ is after $\W{}'$ in $\lS$.
Since $ts < ts'$,
	by~\ref{worder} of Definition~\ref{Lin},
	$\W{}$ is before $\W{}'$ in $\lS$.
Since $\RD{}$ is after $\W{}'$,
	$\W{}$ is before $\RD{}$ in $\lS$.	
\end{proof}

\begin{lemma}\label{rr}
If a read operation $\RD{}$ completes before a read operation $\RD{}'$ starts in $H$ and $\RD{},\RD{}'\in \lS$,
	then $\RD{}$ is before $\RD{}'$ in $\lS$.
\end{lemma}

\begin{proof}
Assume a read operation $\RD{}$ completes before a read operation $\RD{}'$ starts in $H$ and $\RD{},\RD{}'\in \lS$.
Let $(-,ts)$ and $(-,ts')$ denote the values that
	 $\RD{}$ and $\RD{}'$ return, respectively.
Since $\RD{}$ completes before $\RD{}'$ starts,	
	$\RD{}'$ reads all $Val[-]$ after $\RD{}$ reads $(-,ts)$ from $Val[-]$.
Since the timestamps in each $Val[-]$ are non-decreasing,
	$\RD{}'$ reads $(-,ts'')$ from $Val[-]$ for some $ts'' \ge ts$.
By line~\ref{tsmax} of $\Al$,
	$ts'' \le ts'$.
So $ts \le ts'' \le ts'$.

\textbf{Case 1:} $ts = ts'$.
Since $\RD{}$ completes before $\RD{}'$ starts,
	 $\RD{}$ is before $\RD{}'$ in $\RS$.
So by~\ref{rorder3} of Definition~\ref{Lin},
	$\RD{}$ is before $\RD{}'$ in $\lS$.

\textbf{Case 2:} $ts < ts'$.
By Observation~\ref{readinit0},
	 $ts'.sq > 0$.
By~\ref{rorder2} of Definition~\ref{Lin},	
	there is a write operation $\W{}' \in \lS$ that writes $(-,ts')$ to $Val[-]$ 
	such that $\RD{}'$ is after $\W{}'$ in $\lS$.

\textbf{Subcase 2.1}: $ts.sq = 0$.
	By~\ref{rorder1} of Definition~\ref{Lin},	
	$\RD{}$ is before all the write operations in $\lS$.
Then $\RD{}$ is before $\W{}'$ in $\lS$.
So $\RD{}$ is before $\RD{}'$ in $\lS$.

\textbf{Subcase 2.2}: $ts.sq \ne 0$.
By~\ref{rorder2} of Definition~\ref{Lin},
	some write operation $\W{} \in \lS$ writes $(-,ts)$ to $Val[-]$ 
	and $\RD{}$ is after $\W{}$ and before any subsequent write operations in $\lS$.
Since $ts < ts'$,
	by~\ref{worder} of Definition~\ref{Lin},
	$\W{}$ is before $\W{}'$ in $\lS$.
So $\RD{}$ is before $\W{}'$ in $\lS$.
Since $\RD{}'$ is after $\W{}'$ in $\lS$,
	$\RD{}$ is before $\RD{}'$ in $\lS$.
\end{proof}

By Corollary~\ref{ww}, Lemma~\ref{rw}, Lemma~\ref{wr},
	and Lemma~\ref{rr}, we have:
\begin{corollary}\label{l2}
If an operation $o$ completes before an operation $o'$ starts in $H$ and $o,o'\in \lS$,
	then $o$ is before $o'$ in $\lS$.
\end{corollary}

By~\ref{rorder1} and~\ref{rorder2} of Definition~\ref{Lin}, we have:
\begin{observation}\label{l3}
For any read operation $r$ in $\lS$,
		if no write operation precedes $r$ in $\lS$, then $r$ reads the initial value of the register; 
		otherwise, $r$ reads the value written by the last write operation that occurs before $r$ in $\lS$.
\end{observation}

By Lemma~\ref{l1_2}, Corollary~\ref{l2} and Observation~\ref{l3},
	$f(H) = \lS$ satisfies properties~\ref{p1}, \ref{p2}, and \ref{p3} of Definition~\ref{defl}.
So:

\begin{lemma}
$f$ is a linearization function for the set of histories $\seth$ of $\Al$.
\end{lemma}

Thus:

\lampol*

\section{Proof of Theorem~\ref{swequal}}\label{proofsw}

We now prove Theorem~\ref{swequal} in Section~\ref{swsection},
	namely, \emph{any} linearizable implementation of SWMR registers is $\sly$
	(this holds for message-passing, shared-memory, and hybrid systems).
Thus, the well-known ABD implementation of SWMR registers in message-passing systems
	is not only linearizable; it is actually write strongly-linearizable.
(This implementation however is \emph{not} strongly linearizable~\cite{abdnotsl}.)

Consider an arbitrary implementation $\swA$ of a SWMR register.
Let $\seth$ be the set of histories of $\swA$.
Since $\swA$ implements a single-writer register, 
	the following holds:
\begin{observation}\label{swob}
In any history $H\in \seth$,
\begin{enumerate}
	\item \label{swob11}there are no concurrent write operations, and 
	\item \label{swob12}there is at most one incomplete write operation.
\end{enumerate}
\end{observation}

By part~\ref{swob11} of Observation~\ref{swob} and property~\ref{p2} of Definition~\ref{defl},
	we have the following:
\begin{observation}\label{swob2}
For any history $H\in \seth$ and any linearization function $f$ of $\seth$,
	the write operations in $f(H)$ are totally ordered by their start time in $H$.
\end{observation}

\begin{lemma}
If $\swA$ is linearizable,
	then $\swA$ is $\sly$.
\end{lemma}

\begin{proof}
Assume $\swA$ is linearizable.
By Definition~\ref{L-I},
	 there is a linearization function $f$ for $\seth$.
Consider a function $f^*$ that is modified from $f$ as follows:
for any history $H$ and its linearization $f(H)$,
	if the last operation $o$ in $f(H)$ is a write operation that is incomplete in $H$, 
	then we obtain $f^*(H)$ by removing $o$ from $f(H)$;
	otherwise, $f^*(H)$ equals $f(H)$.

\begin{claim}\label{ifwinthencor}
If a write operation $w$ is in $f^*(H)$,
	then $w$ is completed or read by some read operation in $H$.
\end{claim}

\begin{proof}
Assume, for contradiction,
	a write operation $w \in f^*(H)$ is incomplete and not read by any read operation in $H$.
Since $w \in f^*(H)$,
	$w$ is in $f(H)$ such that $w$ is not the last operation in $f(H)$ ($\star$).
Since $w$ is incomplete,
	by Observation~\ref{swob},
	$w$ is the last write operation in $H$.
By Observation~\ref{swob2},
	$w$ is the last write operation in $f(H)$.
Furthermore,
	since $w$ is not read by any read operation in $H$,
	no read operation is after $w$ in $f(H)$.
Thus, $w$ is the last operation in $f(H)$,
	which contradicts ($\star$).
\end{proof}

\begin{claim}\label{ifcorthenwin}
If a write operation $w$ is completed or read by some read operation in $H$,
	then $w$ is in $f^*(H)$,
\end{claim}

\begin{proof}
Assume a write operation $w$ is completed or read by some read operation in $H$.

\textbf{Case 1:} $w$ is completed in $H$.
Since $f$ is a linearization function for $\seth$,
	by property~\ref{p1} of Definition~\ref{defl}, 
	$w$ is in $f(H)$.
Since $f^*(H)$ removes only the incomplete operation from $f(H)$,
	$w$ is in $f^*(H)$.

\textbf{Case 2:} $w$ is read by some read operation $r$ in $H$.
Since $f$ is a linearization function for $\seth$,
	by property~\ref{p3} of Definition~\ref{defl},
	$w$ is before $r$ in $f(H)$ and so $w$ is not the last operation in $f(H)$.
Since $f^*(H)$ removes only the last operation from $f(H)$,
	$w$ is in $f^*(H)$.
	
Thus, in both cases 1 and 2,
	 $w$ is in $f^*(H)$.
\end{proof}

\begin{claim}
$f^*$ is a linearization function of $\mathcal{H}$.
\end{claim}
\begin{proof}
Consider any history $H\in \mathcal{H}$.
Since $f^*(H)$ removes from $f(H)$ only the operation that is incomplete in $H$,
	$f^*(H)$ still satisfies properties~\ref{p1} and \ref{p2} of Definition~\ref{defl}.
Since $f^*(H)$ removes only the last operation from $f(H)$,
	$f^*(H)$ still satisfies property~\ref{p3} of Definition~\ref{defl}.
Therefore, 
	$f^*$ is a linearization function of $\mathcal{H}$.
\end{proof}

\begin{claim}
$f^*$ satisfies property (P) of Definition~\ref{defwsl}. 
\end{claim}
\begin{proof}
Consider histories $G, H\in \mathcal{H}$ such that $G$ is a prefix of $H$.
Let $W_G$ and $W_H$ denote the write sequences in $f^*(G)$ and $f^*(H)$ respectively.
By Claim~\ref{ifwinthencor},
	 all the operations in $W_G$ are completed in $G$ or read by some read operations in $G$.
Since $G$ is a prefix of $H$,
	all operations in $W_G$ are also completed in $H$ or read by some read operations in $H$.
So By Claim~\ref{ifcorthenwin},
	$W_H$ contains all the operations in $W_G$.
By Observation~\ref{swob2},
	the write operations in $W_G$ and $W_H$ are totally ordered by their start time in $G$ and $H$, respectively.
Then since $G$ is a prefix of $H$ and $W_H$ contains all the operations in $W_G$ ,
	$W_G$ is a prefix of $W_H$. 
So $f$ satisfies property (P) of Definition~\ref{defwsl}.
\end{proof}
Thus,
	 $f$ is a $\str$ function for $\seth$ and $\swA$ is $\sly$.
\end{proof}

\swe*

\end{appendices}

\end{document}